\documentclass[12pt]{article}
\usepackage[utf8]{inputenc}
\usepackage{microtype}
\usepackage[margin=1in]{geometry} %
\usepackage{booktabs,caption,subcaption,dcolumn} %
\usepackage[flushleft]{threeparttable} %
\usepackage[table]{xcolor}
\usepackage{footmisc}
\usepackage{adjustbox}
\usepackage{multirow}
\usepackage{multicol}
\usepackage{enumerate} %
\usepackage{siunitx}
\usepackage{tabularx}
\usepackage{pdflscape} 
\usepackage{rotating}

\usepackage{setspace}
\makeatletter
\newcommand*{\centernot}{%
  \mathpalette\@centernot
}
\def\@centernot#1#2{%
  \mathrel{%
    \rlap{%
      \settowidth\dimen@{$\m@th#1{#2}$}%
      \kern.5\dimen@
      \settowidth\dimen@{$\m@th#1=$}%
      \kern-.5\dimen@
      $\m@th#1\not$%
    }%
    {#2}%
  }%
}
\makeatother

\usepackage{amsmath}
\usepackage{amsthm}
\usepackage{amssymb}
\usepackage{amsfonts,bm}
\newcommand{\APO}{\mathrm{APO}}
\newcommand{\ATE}{\mathrm{ATE}}
\theoremstyle{definition}

\theoremstyle{plain}
\newtheorem{theorem}{Theorem}

\newtheorem{lemma}{Lemma}

\newtheorem{proposition}{Proposition}
\newtheorem{appprop}{Proposition}[section]

\theoremstyle{definition}
\newtheorem{assump}{Assumption}

\makeatletter
\newenvironment{assumption}[2][]{%
  \def\assump@old{\theassump}%
  \if\relax\detokenize{#1}\relax\else
    \renewcommand{\theassump}{#1}%
  \fi
  \begin{assump}%
}{%
  \end{assump}%
  \renewcommand{\theassump}{\assump@old}%
}
\makeatother

\usepackage{algorithm}
\usepackage{algpseudocode}

\usepackage[style=apa,natbib=true]{biblatex}
\usepackage{hyperref}
\hypersetup{
    colorlinks=true,
    linkcolor=blue,
    urlcolor=blue,
    citecolor=blue
}

\usepackage[affil-it]{authblk}

\title{Identification of Child Penalties}
\author{Dor Leventer\thanks{University College London. Email: d.leventer@ucl.ac.uk. I am indebted to my PhD Advisor Itay Saporta-Eksten for guidance in this project. For comments I thank Yoav Goldstein, Analia Schlosser, Roee Levy, Oren Danieli, Ran Abramitzky, David Neumark, Netanel Ben-Porath, Jonathan Roth, Jonathan Shaw, and participants in Tel Aviv University, Haifa University, Bar Ilan University, The Hebrew University, Coller School of Management in Tel Aviv University, Ben Gurion University, Reichman University, Arlozorov Forum, the Israeli Economic Association Annual Conference 2026, the World Inequality Conference 2026, IAAE 2026 Annual Conference, EEA-ESEM 2026 Congress, and University College London. I gratefully acknowledge financial support from The Israel Pollak Fellowship Program for Excellence and the Arlozorov Forum Labor Markets Scholarship.}}
\affil{}
\date{\today}

\graphicspath{{./}}

\begin{document}

\maketitle

\begin{abstract}
\noindent
A large literature estimates child penalties using event studies by gender, normalizing by predicted earnings absent children and reporting the gender gap.
I formalize the design's identification assumption: normalized parallel trend violations are equal across genders, which I term normalized triple differences (NTD). 
Under NTD, the conventional target, the gender gap in normalized effects, is not identified when parallel trends fails. 
However, the effect of parenthood on the gender earnings ratio is point identified. 
Using Israeli administrative data, I find that parenthood's share of gender inequality varies by age at first childbirth. 
\medskip
\noindent\textit{Keywords:} child penalty; gender inequality; event study; difference-in-differences; causal inference.

\noindent\textit{JEL codes:} J16, C21, C23, J13.
\end{abstract}

\clearpage

\noindent
The transition to parenthood is central to the onset of gender inequality in labor markets \citep{goldin2024nobel}. 
A large body of research quantifies how the effects of parenthood on labor market earnings differ between women and men \citep[e.g.,][]{kleven2019children}. 
The common approach, which I term normalized event studies, proceeds as follows. 
Within gender, earnings are regressed on event-time indicators and age and calendar-year fixed effects. 
Predicted earnings constructed from the fixed effects alone serve as the estimated counterfactual earnings absent children. 
The estimated event-time coefficients, normalized by the mean of the predicted counterfactual earnings, capture gender-specific percentage effects of parenthood on earnings. 
The gender gap in the normalized estimated effect of parenthood on earnings is then analyzed, and is often termed the ``child penalty''.

This paper studies the identification framework underlying the above between-gender normalized event-study empirical strategy.
Most of the analysis focuses on comparing a single treatment group to a single not-yet-treated control group, for example parents who had their first child at 20 as treated vs. 30 as control.
I begin with within-gender comparisons, and argue that difference-in-differences (DID) estimates are biased due to parallel trends violations, as higher human capital individuals both delay childbirth and have steeper life-cycle earnings trajectories.
I then turn to the between-gender normalized event study approach. 
First, I derive the identification assumptions from the tests used in applied work, which I term normalized triple differences (NTD).
I show that the gender gap in normalized effects, the target causal estimand of normalized event studies, is not identified under NTD when parallel trends is violated. 
A bias-bounding exercise suggests the resulting bias is large for earlier treatment groups. 
I next consider a new causal estimand, the effect of parenthood on the gender earnings ratio, and show it is identified under NTD. Empirically, I document that parenthood accounts for less of gender earnings inequality among older treatment groups. 
Finally, I discuss causal estimands that aggregate across treatment groups, and show how differences in the fertility-timing distribution complicate comparisons of aggregate estimates, such as between countries. 
The theoretical discussion is illustrated empirically using Israeli administrative data.

Considering an event study within a single gender, identification of the effect of parenthood on earnings relies on the parallel trends assumption, namely that the treatment and control groups have the same trends in the counterfactual absent childbirth.
Human capital models of fertility \citep[e.g.,][]{becker1990human} and life-cycle earnings \citep[e.g.,][]{ben1967production} suggest that individuals with higher labor-market ability both delay childbirth and have steeper earnings trajectories. This mechanism predicts that parallel trends are violated: the control group, which by definition had children later and so is, on average, of higher ability, has a steeper counterfactual earnings trend than the treatment group, which had children earlier and so has lower ability. 
I strengthen this argument empirically in two ways. 
First, I show that individuals who delay childbirth have, on average, higher values of pre-childbirth covariates associated with human capital.
Second, I assess pre-trends. I find that within gender, pre-treatment DID estimates increase in magnitude with the gap in treatment timing between treated and control groups.
Taken together, since later event times require increasingly distant not-yet-treated control groups, and more distant control groups differ more in human capital, which introduces larger violations of parallel trends, DID estimates at longer horizons are less credible.

To deal with within-gender selection, researchers may difference between genders, which absorbs parallel trends violations that are common to both genders.
However, as discussed above, current applications normalize within-gender estimates, and then difference between genders.
A rigorous discussion of the identification framework requires articulating the identification assumptions under this approach.
To do so, I consider the between-gender validation test with which researchers give causal interpretation to their estimates: namely, that the pre-treatment gender gap in normalized estimates is zero.\footnote{For a partial survey of the literature and the wording of the validation test, see Appendix~\ref{sec:app_validation_lit}. I find that papers support a causal interpretation with between-gender parallel pre-trends, and not with within-gender zero pre-trends.}
I show that if no anticipation holds, this test validates the following restriction: parallel trends violations, divided by counterfactual earnings, are equal across genders, which I refer to as NTD. In other words, NTD states that the percentage error in imputing counterfactual earnings is the same for women and men, a percentage analog of the triple differences (TD) assumption that parallel trends violations in levels are equal across genders.

Having articulated NTD, I show that even when NTD holds, the descriptive gender gap in normalized estimates does not identify its target, the causal gender gap in normalized effects, when the parallel trends assumption is violated.
The resulting bias is multiplicative and corresponds to the percentage error in imputing counterfactual earnings. For example, if parallel trends violations lead to over-imputing counterfactual earnings by 10\%, the gender gap in normalized effects is scaled by $1/1.1 \approx 0.91$, i.e., its magnitude understated by 9\%. In general, if the control group has steeper counterfactual earnings growth than the treated group, the descriptive gender gap understates its target causal estimand.
To quantify the bias, I develop a bias-bounding exercise, using a bias-correction formula that assumes fathers' effects are known. Empirically, the exercise suggests substantial bias for earlier treatment groups, with the conventional estimator understating child penalties. For example, for parents that had their first child at age 26, bias-corrected estimates are 24--51\% more negative than conventional estimates five years post-childbirth.

Turning to the identification result, I show that in treatment groups where NTD holds, the effect of parenthood on the gender earnings ratio is point identified. To illustrate the difference between the two targets, suppose that absent childbirth women would earn 80\% of men, and that childbirth reduces women's earnings by 50\% and men's by 10\%. The gender gap in normalized effects is $-40$ percentage points ($-50\%$ minus $-10\%$). The effect on the gender earnings ratio is $-36$ percentage points: childbirth moves the ratio from 80\% to $80\%\times0.5/0.9\approx44\%$. The latter answers directly how much of observed gender inequality is caused by parenthood. I develop a corresponding estimator and derive cluster-robust standard errors based on its influence function.
The new estimand also enables a direct decomposition of the gender gap. A novel empirical finding is that the extent to which parenthood accounts for gender inequality is heterogeneous by age of first childbirth: at childbirth, parenthood's share of observed gender earnings inequality ranges from 84\% to 59\% for parents who had their first child at ages 26 and 30, respectively.

The analysis up to this point focuses on a single treatment group. Aggregation across multiple treatment groups introduces a subtlety: I show that the conventional aggregate estimator implicitly weights treatment groups by their counterfactual earnings, hence giving more weight to later-treated, higher-earning groups. I argue for an aggregate that weights by the treatment distribution alone. Turning to comparisons of aggregated effects across strata, such as between countries, I argue that differences in aggregate estimates confound differences in treatment-group-specific effects with differences in fertility-timing distributions. I illustrate this empirically by using the Israeli treatment-group-specific estimates and varying the treatment distribution, and show that the fertility-timing distribution alone can determine whether the aggregate effect of parenthood on gender inequality appears to recede.

\paragraph{Literature review.} This paper is most closely related to studies of child penalties that use between-gender normalized event studies as their empirical strategy. Appendix \ref{sec:app_validation_lit} presents a non-exhaustive list of papers. 
I make four contributions to this literature. First, I articulate the identification framework as NTD. Second, I show that under NTD the conventional target estimand is generally not identified. Third, I show that the effect of parenthood on the gender earnings ratio is point identified under NTD. Fourth, I discuss why the aggregate event-study validation test common in applied work does not test the identification assumption, and propose validation tests conducted for a single treatment–control pair. Beyond child penalties, the NTD framework can be applied to other settings that aim to identify the effect of a treatment on between-group inequality, and remains well-defined when the outcome has many zeros. An open-source \textbf{\textsf{R}} package, \href{https://dorleventer.github.io/childpen}{\texttt{childpen}}, implements the proposed estimator.

This paper also relates to within-gender estimators, which are not subject to the biases due to normalization. As examples, \citet{melentyeva2023child} implement a stacked DID \citep{cengiz2019effect,wing2024stacked}, \citet{thakral2026child} include cohort fixed effects alongside transparent identifying restrictions, \citet{lin2025long,fajardo2024there} use the estimator of \citet{sun2021estimating}, and \citet{bearth2024beyond} adopt the estimator of \citet{callaway2021difference}.
While these methods correct estimation and specification biases, they do not address identification bias due to violations of the parallel trends assumption.
I argue such violations are likely when the outcome is labor market earnings, based on human capital models.
Consistent with this argument, \citet{bensnes2023reconciling} and \citet{ilciukas2025parenthood} exploit quasi-random success of fertility treatments in Norway and the Netherlands, respectively, and document that event-study estimates are biased by selection into the timing of parenthood. This paper differs from theirs in that I focus on the identification underlying the between-gender normalized approach, while their target is valid identification of within-gender effects.
For other outcomes where parallel trends may be credible, a further advantage of the proposed approach is that the analogous gender-ratio estimand is also identified under DID. Hence, child penalty studies of such outcomes can report it directly, without adopting NTD. More broadly, in other DID applications studying between-group inequality, such an estimand is identified and serves as a potential target.

This paper also connects to the broader DID methodological literature on parallel trends testing \citep{roth2022pretest,ghanem2022selection,roth2023parallel} and DID as an aggregate of multiple $2\times2$s \citep{callaway2021difference,goodman2021difference}.
On validation, I argue that when the control group changes across post-treatment event time, as in staggered designs that use not-yet-treated units as controls, the identifying assumption must be tested separately for each control group, or for each composition of control groups.
On aggregation, I analyze new aggregate estimands and estimators and discuss the implications of treatment-distribution heterogeneity for cross-country and cross-subgroup comparisons \citep[e.g.,][]{kleven2019child,andresen2022causes}.

\paragraph{Roadmap.} The rest of the paper is organized as follows.
Section~\ref{sec:empirical_context} presents the normalized event-study empirical strategy, the Israeli data, and the identification setup.
Section~\ref{sec:did_violation} argues that within-gender DID estimates are biased due to selection into the timing of parenthood.
Section~\ref{sec:iden_ass} articulates NTD and establishes that the gender gap in normalized effects is not identified under NTD.
Section~\ref{sec:alt_estimand} shows that the effect of parenthood on the gender earnings ratio is point identified under NTD.
Section~\ref{sec:aggregation} discusses aggregation across treatment groups. 
Section~\ref{sec:conclusion} concludes.

\section{Empirical Context and Identification Setup}\label{sec:empirical_context}
This section formalizes normalized event studies,
describes the Israeli administrative data used to illustrate the theory, and presents the notation used throughout the paper.

\subsection{Normalized Event Studies}\label{sec:event_studies}
Let $Y_{i,a}$ denote real annual labor-market earnings of individual $i$ at age $a$, and let $D_i$ denote the age at which individual $i$ has their first child. Define $E_{i,a} = a - D_i$ as the time since first childbirth, which I term event time. Let $G_i \in \{f, m\}$ indicate gender, where $f$ represents female and $m$ male.
The estimation algorithm for the normalized event-study proceeds in three steps.
First, outcomes are regressed on event-time indicators and fixed effects, separately for each gender. The regression model for gender $g \in \{f, m\}$ is
\begin{equation}\label{eq:event_study}
Y_{i,a} = \sum_{e \neq -1} \beta_e^g 1_{\{E_{i,a} = e\}} + \alpha_a^g + \alpha_t^g + u_{i,a},
\end{equation}
where $1_{\{\cdot\}}$ denotes an indicator function, $\alpha_a^g$ and $\alpha_t^g$ are age and year fixed effects, respectively, and the superscript $g$ indexes gender-specific coefficients.
In the second step, predicted earnings net of event-time coefficients are computed as $\widetilde{Y}_{i,a} = \widehat{\alpha}_{a}^{G_i} + \widehat{\alpha}_{t}^{G_i}$. 
Finally, the estimated event-time coefficients $\widehat{\beta}_e^g$ are normalized by the conditional mean of $\widetilde{Y}_{i,a}$ within gender and event time, where $\mathbb{E}_n[\cdot]$ denotes the sample mean:
\begin{equation}\label{eq:cp_main_res} 
\widehat{\theta}_{\mathrm{ES}}(g,e)
= \frac{\widehat{\beta}_e^g}{\mathbb{E}_n\left[\widetilde{Y}_a \mid G = g, E_a = e\right]}.
\end{equation}

Recent work has shown regressions similar to \eqref{eq:event_study} can produce biased estimates in the presence of multiple treatment groups \citep{sun2021estimating,goodman2021difference,de2020two,borusyak2024revisiting}. As \citet{melentyeva2023child} argue, this concern applies to the child penalty setting as well. My focus is on biases arising from the identification framework, which remain even when comparing a single treatment group to a single control group across two time periods, i.e., a single $2\times2$.

\subsection{Data}\label{sec:data}
Although the arguments developed below are theoretical and generalizable, I illustrate their implications using a specific application to aid with constructing intuition, empirically assess key claims, and document new empirical insights. To that end, I describe here the main Israeli administrative data used throughout the paper. To assess generalizability, UK and German survey data are also used, with details provided in Appendix~\ref{sec:app_survey_data}.
The raw dataset covers all Israeli citizens born between 1970 and 2000, matched to their spouses, parents, and children. It was compiled by the Israeli Central Bureau of Statistics (CBS) and integrates data from several administrative sources, including the Civil Registry, the Ministry of Education and the Israeli Income Tax Authority.

\paragraph{Main Variables.} I now turn to formally defining the treatment and outcome variables. 

\textit{Treatment: Age at birth of first child.} Each individual is linked to their biological children. The year of birth of the earliest child is used to define the year of first childbirth. Subtracting the parent’s year of birth yields their age at birth of first child.

\textit{Outcome: Earnings.} Annual labor market earnings are observed from 1999 to 2020, based on micro-level tax records. Earnings are coded as zero in years with no reported income. All values are expressed in real 2020 New Israeli Shekels (NIS), using the CBS CPI.

The analysis below makes use of several additional variables, defined explicitly in Appendix \ref{sec:appendix_data_vars}. These include grandparents' earnings, nationally administered test scores called Meitzav, number of credits in high-school subjects, years of education and highest education degree. The sample definition also uses ethnicity and religion variables.

\paragraph{Sample Definitions.}
I make the following restrictions on the main sample. First, individuals identified as Arab or Ultra-Orthodox (Haredi) Jews are excluded, due to systematically different fertility and labor market trajectories \citep{yakin2021,gould2024child}. I restrict the sample to individuals born between 1975 and 1990 due to limitations in time coverage of the earnings data.
Furthermore, I drop individual-year observations where the individual is younger than 20, corresponding to the typical entry into the labor market after high school completion and mandatory army service, and years in which the individual is reported as dead.

I keep individuals in treatment groups $D=20,\ldots,40$, and focus in the main analysis on the subset $D=24,\ldots,34$. The lower bound reflects the requirement that pre-treatment diagnostics extend four years before treatment. The upper bound is set so that estimates can be reported up to five post-treatment years: a parent in $D=34$ is at age 39 five years post-treatment, and the nearest not-yet-treated control then comes from $D=40$. Groups past 40 have very small sample sizes, and hence are excluded.
The dataset used in the main analysis, after the above limitations, consists of 13.7 million individual-year observations, made up of 374 thousand mothers and 320 thousand fathers.
Further construction details are provided in Appendix~\ref{sec:appendix_data_def}.

\paragraph{Sample Statistics.}
Figure \ref{fig:stats} documents for the Israeli data the treatment distribution (panel (a)), mean annual earnings by treatment and gender (panel (b)), and normalized event study estimates of~\eqref{eq:cp_main_res} (panel (c)).
The treatment mode is approximately 28--29 for mothers and 30--31 for fathers.
In panel (b), a visible drop in mean earnings at the age of first childbirth can be observed for women, but not for men. However, understanding what happens counterfactually several years after first childbirth is not immediate from such sample statistics, motivating the event-study based estimates, reported in panel (c). The normalized estimates for mothers and fathers, $\widehat{\theta}_{\mathrm{ES}}(f,e)$ and $\widehat{\theta}_{\mathrm{ES}}(m,e)$ from \eqref{eq:cp_main_res}, are very similar before childbirth. After childbirth, a gender gap emerges, with mothers’ normalized effects more negative than fathers’.

\begin{figure}[p]
    \centering
    \begin{subfigure}[b]{\textwidth}
        \centering
        \subcaption{Distribution of Age at First Childbirth}
        \label{fig:d_dist_israel}
        \includegraphics[width=\textwidth]{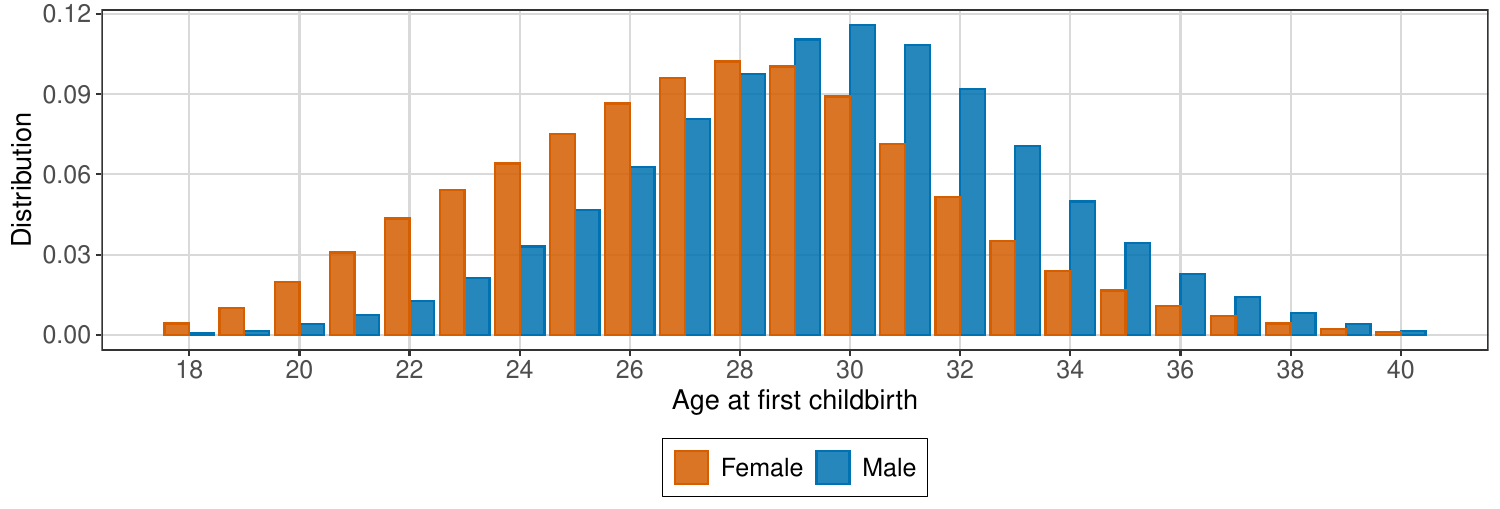}
    \end{subfigure}

    \begin{subfigure}[b]{\textwidth}
        \centering
        \subcaption{Mean Annual Earnings by Age at First Childbirth and Gender}
        \label{fig:raw_earnings_by_d}
        \includegraphics[width=\textwidth]{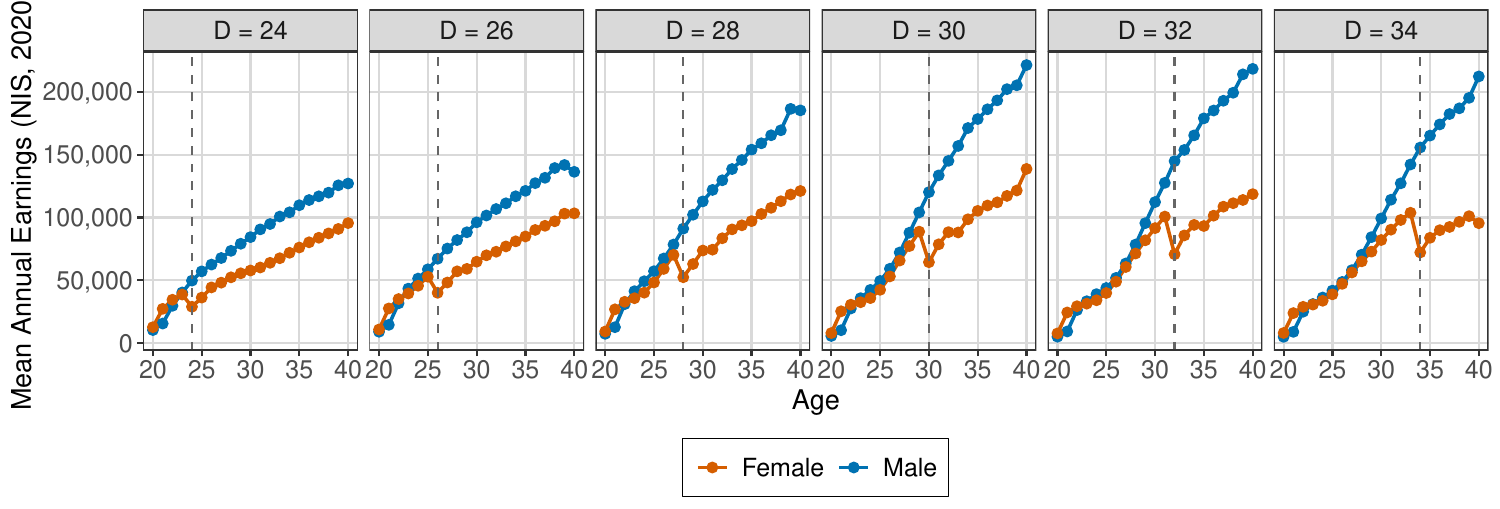}
    \end{subfigure}

    \begin{subfigure}[b]{\textwidth}
        \centering
        \subcaption{Normalized Event Studies}
        \label{fig:event_study}
        \includegraphics[width=\textwidth]{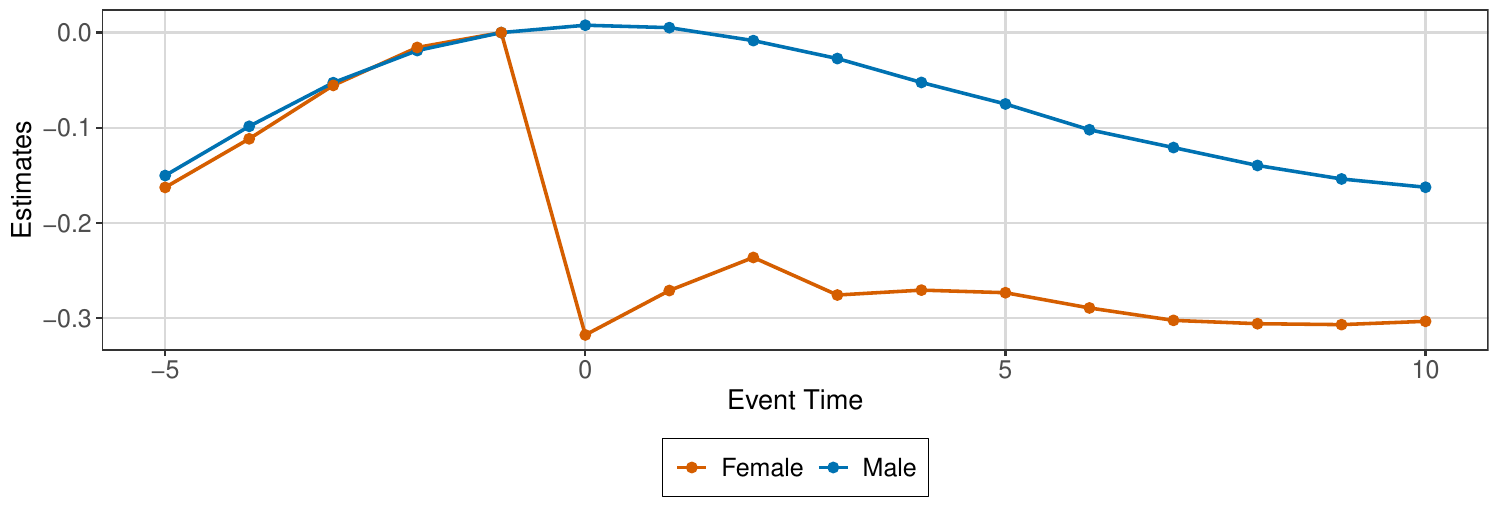}
    \end{subfigure}

    \caption{Sample Statistics and Event Studies. \textit{Notes:} Based on the Israeli administrative data. Panel~(a) documents the distribution of age at first childbirth by gender (colors). Panel~(b) documents mean annual earnings by age for selected treatment groups (columns), separately by gender. Vertical dashed lines mark the age at first childbirth. Panel~(c) presents normalized-event-study child penalty estimates. The sample in panel~(c) includes individuals observed in all ages from 5 years prior to treatment to 10 years post treatment.}
    \label{fig:stats}
\end{figure}

\subsection{Identification Setup}\label{sec:iden_setup}
In this section I lay out the necessary setup for the discussion on identification. 

\paragraph{$2\times2$ comparisons.}
For a given gender $g$, consider a comparison between a treatment group $d$ and a control group $d^\prime$, and between a target age $a$ and a pre-treatment age $d-1$. Following \citet{melentyeva2023child}, I set the control group to be $d^\prime=a + 1$, the closest-not-yet-treated treatment group.
For example, if $a = 30$ then $d^\prime = 31$, if $a = 31$ then $d^\prime = 32$, and so on.
In event-study terms, the target age $a$ corresponds to event time $e = a - d$, so the $2\times2$ objects defined below are the building blocks of the event-study estimator in~\eqref{eq:cp_main_res}.
I introduce each object below with its full arguments and then suppress $(d,d^\prime,a)$ for brevity.

\paragraph{Potential Outcomes.}
Recall $D_i$ is the age at first childbirth. 
Let $W_{i,a} = 1_{\{a \geq D_i\}}$ denote the treatment status of individual $i$ at age $a$.
This treatment definition fits a staggered adoption design, i.e., $W_{i,a-1}=1\rightarrow W_{i,a}=1$. 
Hence, we can treat parents with $D_i=d$ as a distinct treatment group, untreated in ages before $d$ and treated from age $d$ and onward. 

Under the stable unit treatment value assumption (SUTVA) \citep{rubin1980randomization}, in a staggered adoption design potential outcomes are a function of the timing of treatment \citep[see, e.g.,][]{callaway2021difference}.\footnote{This requires that “age at first birth” satisfies the assumption known as treatment variation irrelevance \citep{vanderweele2009concerning}, one of the two elements in SUTVA. Since this paper's contribution does not concern relaxing SUTVA, but rather analyzing parallel-trend type assumptions, I maintain this assumption throughout.} Formally, let $Y_{i,a}(d)$ be the potential outcome of individual $i$ at age $a$ if the first childbirth occurs at age $d$. Let $Y_{i,a}(\infty)$ denote the potential outcome if $i$ never has a child. Observed outcomes are linked to potential outcomes by the consistency assumption: $Y_{i,a} = Y_{i,a}(D_i)$.

\paragraph{Causal Estimands.}
Let $\APO(g, d, d', a) = \mathbb{E}[ Y_a(d') \mid G = g, D = d]$  
denote the average potential outcome (APO) at age $a$ for individuals of gender $g$ who had their first child at age $d$, had they instead had their first child at age $d'$.  
For a single $(d,a)$, denote $\APO(g,d)\equiv \APO(g, d, d, a)$ and $\APO(g,\infty)\equiv\APO(g, d, \infty, a)$; here the retained argument denotes the counterfactual first-birth age, with the treatment group $d$ and target age $a$ suppressed, so $\APO(g,d)$ is the APO under the observed treatment and $\APO(g, \infty)$ is the APO under the counterfactual treatment of never having children.
Next, define the average treatment effect (ATE) for gender $g$, treatment group $d$ at age $a$ as  $\ATE(g, d, a) = \APO(g, d) - \APO(g, \infty)$. For a single $(d,a)$, denote $\ATE(g)\equiv \ATE(g, d, a)$.
To mirror the event study estimator $\widehat{\theta}_{\mathrm{ES}}(g, e)$ in \eqref{eq:cp_main_res}, I define the causal estimand $\theta(g, d, a) = \ATE(g)/\APO(g, \infty)$.
$\theta(g, d, a)$ captures the proportional earnings loss from childbirth, relative to the counterfactual of never giving birth, at age $a$ for individuals of gender $g$ from treatment group $d$. For a single $(d,a)$, denote $\theta(g)\equiv\theta(g, d, a)$.\footnote{This interpretation aligns with the estimand targeted in child penalty studies, i.e., $\widehat{\theta}_{\mathrm{ES}}(g,e)$ in \eqref{eq:cp_main_res}. For example, \citet[p.~188]{kleven2019children} describe their child penalty estimator as “the year-$t$ effect of children as a percentage of the counterfactual outcome absent children,” where $t$ corresponds to event time $e$ in my notation. A similar estimand is discussed in \citet{ciani2019dif}, and it is also used in the vaccine efficacy and excess mortality literatures \citep{orenstein1985field,msemburi2023estimates}.}  

The term “child penalty” is frequently used to describe the differential impact of parenthood on labor market outcomes between women and men. The normalized event study approach compares $\widehat{\theta}_{\mathrm{ES}}(g,e)$ in \eqref{eq:cp_main_res} across gender. Hence $\theta(f) - \theta(m)$, which captures the gender gap in relative earnings losses from parenthood, represents the child penalty target causal estimand.

\paragraph{Descriptive Estimands.}
By a descriptive estimand, I mean a population expectation defined solely in terms of observed outcomes and covariates, without involving potential outcomes beyond the realized outcome $Y$ \citep{abadie2020sampling}. The following three descriptive estimands can be thought of as $2\times2$ plims of DID estimators for the counterfactual APO, ATE and $\theta$.
\begin{align}\label{eq:desc_est}
\delta_{\mathrm{APO}}(g,d,d^{\prime},a) & =\mathbb{E}[Y_{d-1} \mid G = g, D = d] + \mathbb{E}[Y_a - Y_{d-1} \mid G = g, D = d'], \nonumber\\
\delta_{\mathrm{ATE}}(g,d,d^{\prime},a) & = \mathbb{E}[Y_a \mid G = g, D = d] - \delta_{\mathrm{APO}}(g,d,d^{\prime},a), \nonumber\\
\delta_{\theta}(g,d,d^{\prime},a) & = \frac{\delta_{\mathrm{ATE}}(g,d,d^{\prime},a)}{\delta_{\mathrm{APO}}(g,d,d^{\prime},a)}.
\end{align}
As before, for brevity, for a single $(d,a)$ denote $\delta_{\mathrm{APO}}(g)=\delta_{\mathrm{APO}}(g,d,d^{\prime},a)$, $\delta_{\mathrm{ATE}}(g)=\delta_{\mathrm{ATE}}(g,d,d^{\prime},a)$ and $\delta_{\theta}(g)=\delta_{\theta}(g,d,d^{\prime},a)$. In $\delta_{\mathrm{APO}}(g)$, the first term provides the pre-treatment level from the treated group, and the second term adds the trend from the control group. Hence, $\delta_{\mathrm{APO}}(g)$ is how DID constructs the counterfactual APO for the treatment group. $\delta_{\mathrm{ATE}}(g)$ is the conventional DID descriptive estimand: four expectations and three differences. $\delta_{\theta}(g)$ is the ratio of these two, and hence can be viewed as a normalized DID. 
\section{Biases in DID due to Selection}\label{sec:did_violation}
As is well known, parallel trends within gender, defined formally below, allows identification of $\APO(g,\infty)$ and $\ATE(g)$, and in turn of $\theta(g)$.
In this section, I argue that parallel trends is unlikely to hold in the child penalty context due to selection on treatment. A sketch of the argument is that treatment groups differ in human capital, and these differences should affect their earnings trajectories in the absence of treatment. This makes parallel trends less likely to hold as the treatment-age gap between treatment and control groups widens, making within-gender DID estimators less credible in the long run.
Formally, let 
\begin{align*}
\gamma_{\mathrm{PT}}(g, d, d', a) &= \APO(g, d, \infty, a) - \APO(g, d, \infty, d - 1) \\
&\quad - \left[ \APO(g, d', \infty, a) - \APO(g, d', \infty, d - 1) \right],
\end{align*}
denote the difference in counterfactual earnings trends from age $d - 1$ to $a$ between treatment group $d$ and control group $d'$. For brevity, denote $\gamma_{\mathrm{PT}}(g)\equiv\gamma_{\mathrm{PT}}(g, d, d', a)$. Parallel trends within gender can be formalized as follows.
\begin{assumption}[DID-PT]{Parallel Trends}\label{A.pt}
For gender $g$ and a triplet $(d,d',a)$ with $d\leq a<d'$, $\gamma_{\mathrm{PT}}(g)=0$.
\end{assumption}

\subsection{Theoretical Argument}
Standard economic models argue life-cycle earnings trajectories \citep[e.g.,][]{ben1967production} and the timing of childbirth \citep[e.g.,][]{becker1990human} are functions of human capital. These two strands combine in life-cycle models of fertility and labor supply \citep[e.g.,][]{moffitt1984profiles, blackburn1993fertility, adda2017career,eckstein2019career}.

To see the implications for parallel trends, consider two statements.  First, individuals with higher labor-market ability invest more in human capital and delay childbirth. Second, higher-ability individuals have steeper counterfactual earnings trajectories, particularly early in their careers. Combining these two mechanisms generates parallel trends violations in certain $2\times2$s $(d,d^\prime,a)$. For early treatment groups ($d$), the relevant post-treatment ages span the mid twenties to early thirties ($a$). Not-yet-treated control groups ($d^\prime$) are higher-ability individuals who, even absent children, would be on steeper earnings trajectories. Consequently, the parallel trends violation is negative: $\gamma_\mathrm{PT}(g,d,d',a) < 0$ for these specific comparisons. 

\subsection{Empirical Evidence}
Figure \ref{fig:desc_all} documents that in Israel parents who delay childbirth to around age 30 come from higher-earning and more-educated families, score higher on national exams, and are more likely to take advanced tracks in high school.
Appendix~\ref{sec:app_hc_evidence} reports similar patterns for the UK and Germany, suggesting these patterns hold outside the Israeli  institutional setting.\footnote{\citet{melentyeva2023child} show a similar correlation between treatment and grandparents' education in Germany, and \citet{jensen2024birth} show similar patterns in Denmark using the parents' final educational attainment.}

\begin{figure}[p]
    \centering
    \includegraphics[width=\textwidth]{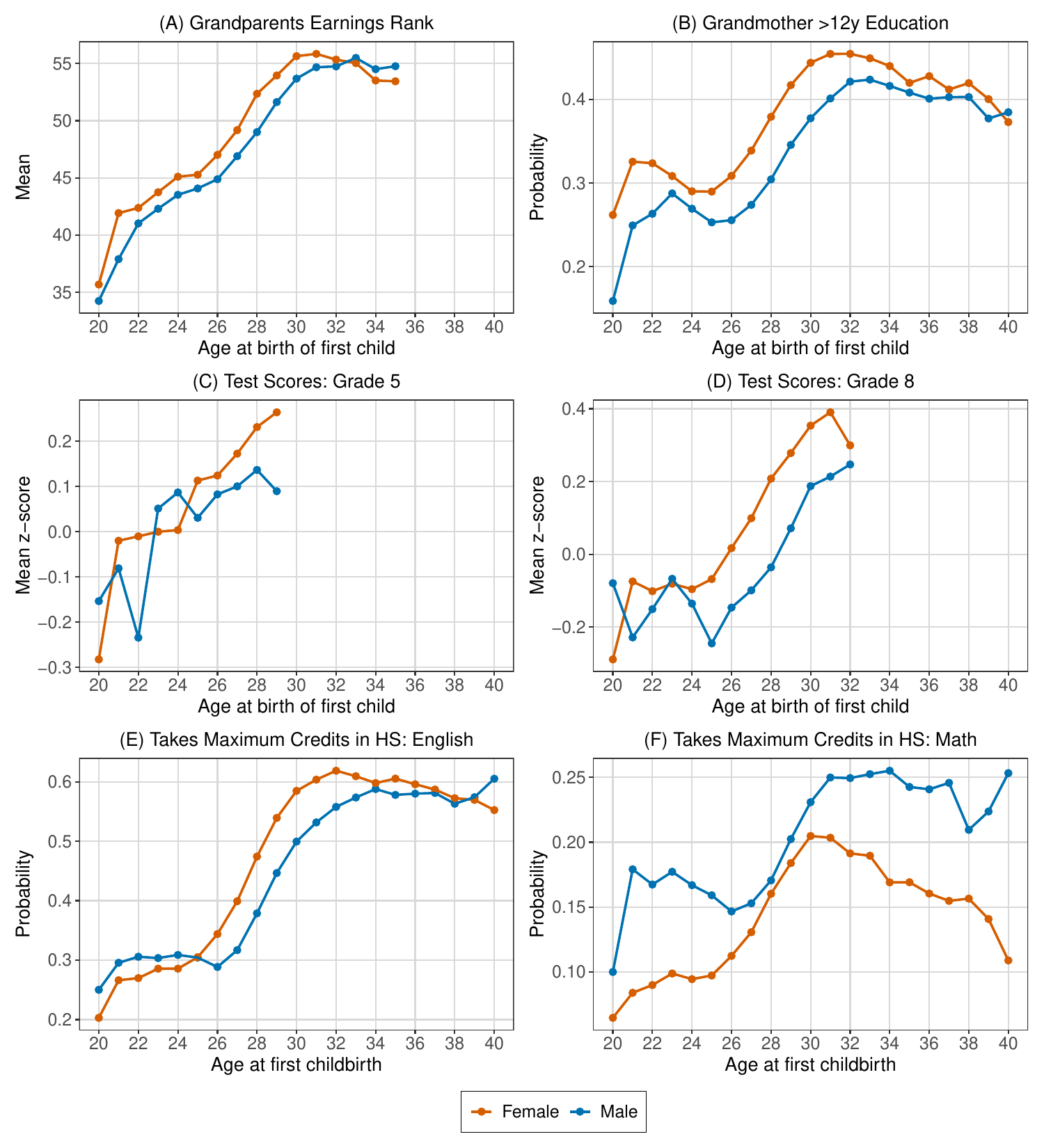}
    \caption{Human Capital and Age at Birth of First Child. \textit{Notes:} The figure presents means of ranked grandparents' earnings when the parents were aged 5--10 (Panel A), probability of grandmothers' education being greater than 12 years (Panel B), means of normalized mathematics test scores in grades 5 and 8 (Panels C and D, respectively), and probabilities of taking the maximum number of credits in English and mathematics in high school (Panels E and F, respectively), by age at birth of first child (x-axis) and gender (colors). The test scores are from a nationally administered test (Meitzav). The maximum number of credits per subject is 5 units. The sample changes by the considered variable due to different data constraints, discussed in detail in Appendix \ref{sec:appendix_data}.}
    \label{fig:desc_all}
\end{figure}

This evidence is in line with selection on treatment.\footnote{Prior studies document the correlation between delayed childbearing and human capital observables using post-treatment variables such as earnings or final education \citep[e.g.,][]{buckles2008understanding}. Using early-life measures minimizes concerns about reverse causality.} The selection pattern that emerges from the Israeli data shows the selection gradient is largest between ages 20–30, then flattens or declines beyond age 30. Selection patterns for fathers peak slightly later than for mothers, consistent with assortative matching and the average two-year spousal age gap. These patterns support the theoretical prediction that parallel trends violations are likely for comparisons between earlier treatment groups (low selection) and later treatment groups (high selection).

\subsection{Validation Tests}\label{sec:violations_validation_did}

In DID, the conventional validation test for Assumption~\ref{A.pt} estimates the difference in trends in pre-treatment ages, commonly referred to as ``pre-trends'' \citep{autor2003outsourcing,roth2022pretest}. As is well known, this is only a suggestive test of Assumption~\ref{A.pt}, which concerns post-treatment ages. Setting the control group $d^\prime$ to the closest-not-yet-treated implies the control group changes in each post-treatment age $a$. In the analysis below (Section~\ref{sec:alt_estimand}) I estimate child penalties on $2\times2$ comparisons ranging across event time $0, \ldots, 5$, requiring six distinct control groups: $d^\prime-d=1,\ldots,6$ for a fixed $d$. Since Assumption~\ref{A.pt} must hold separately for each treatment–control pair, pre-trends testing should be performed at this level as well.\footnote{This point generalizes beyond the closest-not-yet-treated assignment. Whenever the control group consists of not-yet-treated individuals in a staggered adoption design, its composition changes at each post-treatment event time, requiring separate validation tests for each pairing.}
This procedure is not currently done in the literature: event-study validation tests commonly aggregate across all treatment and control groups (Section~\ref{sec:event_studies}). However, such aggregation does not validate Assumption~\ref{A.pt}, which must hold separately for each treatment–control pair.

Figure~\ref{fig:validation_pretrends} reports pre-trends, estimates of $\delta_{\mathrm{ATE}}(g)$ in pre-treatment ages using sample analogs, by treatment–control pairs. 
The figure shows pre-treatment DID estimates increase in magnitude as the treatment–control age gap widens, across nearly all treatment groups and for both mothers and fathers. 
For example, when $D=26$ serves as the control group for treatment group $D=25$ (control $+1$ in the legend) the pre-trends are small with confidence intervals that include zero, while when $D=31$ (control $+6$) serves as control the pre-trends are large and significant. Again, these results support the claim that, as treatment timing between treatment and control groups drifts farther apart, the parallel trends assumption is more likely to be violated. 

\begin{figure}[t!]
    \centering
    \includegraphics[width=\textwidth]{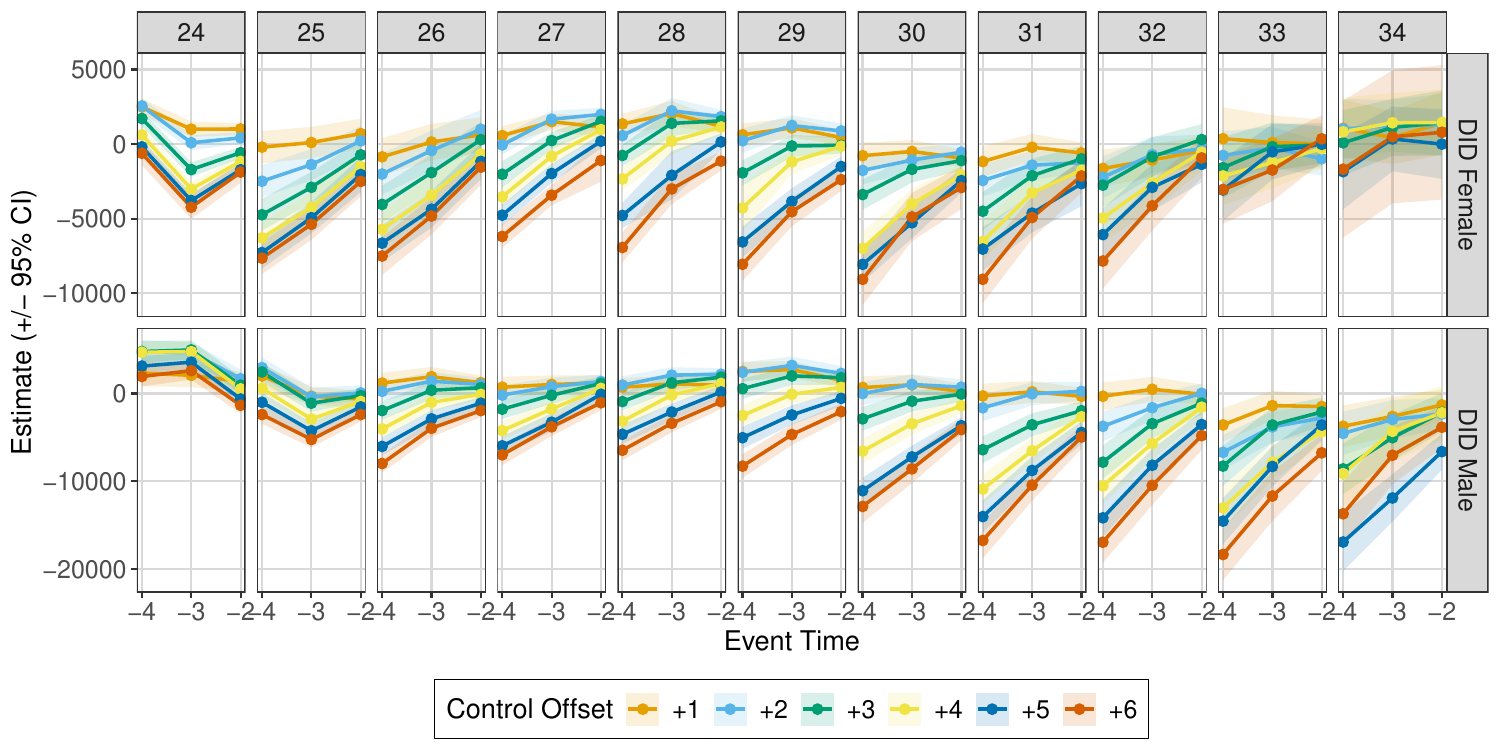}
    \caption{DID Validation Tests by Treatment-Control Pairs. \textit{Notes:} The figure presents pre-treatment DID estimates, with treatment groups in columns, control groups in colors, mothers in the top row and fathers in the bottom row. Control groups $d^\prime$ are shown by distance from treatment group $d$, from $d^\prime-d=1$ up to $d^\prime-d=6$. For example, for treatment group $D=24$ (left-most column), control group $+1$ corresponds to $D=25$, $+2$ to $D=26$, and so on.}
    \label{fig:validation_pretrends}
\end{figure}

\section{Biases in TD due to Normalization}\label{sec:iden_ass}
The previous section argued that within-gender DID is biased due to violations of Assumption \ref{A.pt}.
A potential way forward, since the selection gradient is similar for men and women, is to additionally difference between men and women, i.e., triple difference (TD).
Section~\ref{sec:event_studies} discusses how, in practice, researchers normalize within-gender estimates by counterfactual earnings, and consider differences by gender after normalization.\footnote{Normalizing and then taking the gender difference is similar to transforming to logs and taking the gender difference.
Hence, normalizing may be warranted if researchers believe the gender difference in parallel trend violations is smaller in percentage terms than in levels, similar to the argument to prefer a logarithmic functional form of parallel trends over parallel trends in levels.
Appendix~\ref{sec:td_violation} introduces proportional parallel trend assumptions, building on \citet{ciani2019dif} and \citet{wooldridge2023simple}, formulates their triple difference counterparts, and compares the resulting estimates to those under the normalized triple difference framework discussed in Section~\ref{sec:alt_estimand}.}
This section begins by articulating the identification assumptions that are being tested by normalized event studies. I then show that the conventional target estimand, the gender gap in normalized effects, is not identified even when these assumptions hold. I conclude with a bias-bounding exercise.

\subsection{Deriving the Identification Assumptions}
I now turn to derive the identification assumption that is implicitly maintained in normalized event studies, as inferred from the empirical validation test commonly used in applied work. Before presenting the result, I introduce the no anticipation assumption.  
The no anticipation assumption requires that potential outcomes before childbirth are the same under the observed treatment path and the counterfactual of never giving birth \citep{abbring2003nonparametric}.\footnote{If outcomes at age $d - 1$ are affected by anticipatory behavior, the no anticipation assumption can instead be imposed at earlier ages (e.g., $d - 2$ or $d - 3$), shifting the baseline period for the treated group. Such concerns may also motivate shifting the closest-not-yet control group to later ages (e.g., $a+2$ or $a+3$).} Formally,
\begin{assumption}[NA]{No Anticipation}\label{A.no.ant.1}
For gender $g$, treatment age $d$, and target age $a<d$, $\APO(g, d) = \APO(g, \infty)$.
\end{assumption}
\noindent
Next, I turn to analyzing the validation test used in applied work, that prior to childbirth the gender gap in normalized estimated effects is zero. The null hypothesis can be stated as
\begin{equation}\label{eq:event_study_validation}
    e<0:\theta_{\mathrm{ES}}(f, e) - \theta_{\mathrm{ES}}(m, e) = 0,
\end{equation}
where $\theta_{\mathrm{ES}}(g,e)$ denotes the plim of $\widehat{\theta}_{\mathrm{ES}}(g,e)$ from \eqref{eq:cp_main_res}.
The question is what assumptions this test is validating.
To answer this, I consider the equivalent of \eqref{eq:event_study_validation} for a single $2\times2$:
\begin{equation}\label{eq:ntd_validation_test}
a<d-1, d<d^\prime:\delta_{\theta}(f,d,d^\prime,a)-\delta_{\theta}(m,d,d^\prime,a)=0,
\end{equation}
where $\delta_{\theta}(g)$ is the normalized DID estimand defined in \eqref{eq:desc_est}.
The following result shows what restriction on potential outcomes arises if the test \eqref{eq:ntd_validation_test} and Assumption \ref{A.no.ant.1} hold in the data. 

\begin{proposition}\label{prop:cp_ident_ass}
Consider a $2\times2$ with treatment group $d$, control group $d^\prime > d$, and pre-treatment target age $a < d - 1$.
If Assumption~\ref{A.no.ant.1} holds, then
\begin{equation*}
    \delta_{\theta}(f)=\delta_{\theta}(m)\quad\Longleftrightarrow\quad\frac{\gamma_{\mathrm{PT}}(f)}{\APO(f,\infty)} = \frac{\gamma_{\mathrm{PT}}(m)}{\APO(m,\infty)}.
\end{equation*}
\end{proposition}
\noindent
The proof is presented in Appendix~\ref{sec:app_ident_proofs}.
Proposition~\ref{prop:cp_ident_ass} shows that \eqref{eq:ntd_validation_test} is a joint test of two restrictions: Assumption~\ref{A.no.ant.1} and that the violations of parallel trends, once normalized by the counterfactual APO, are equal across genders. If both hold, the test is satisfied, so a rejection implies at least one is violated.
I now state the additional restriction formally as a new identification assumption:
\begin{assumption}[NTD-PT]{Equal Difference in Normalized Trends}\label{A.ntd_pt}
For a triplet $(d,d',a)$ with $d\leq a<d'$, $\APO(f,\infty),\APO(m,\infty)\neq0$ and
$$\frac{\gamma_{\mathrm{PT}}(f)}{\APO(f,\infty)}=\frac{\gamma_{\mathrm{PT}}(m)}{\APO(m,\infty)}.$$
\end{assumption}
\noindent
An equivalent reading of the assumption is that the percentage error in imputing counterfactual earnings is the same across genders, since the DID-imputed counterfactual differs from the true counterfactual APO by exactly $\gamma_{\mathrm{PT}}$ (see Lemma~\ref{lemma:desc_tie_cause} below). To aid intuition, compare the new assumption to the TD assumption.
\begin{assumption}[TD-PT]{Equal Difference in Trends}\label{A.td}
For a triplet $(d,d',a)$ with $d\leq a<d'$, $\gamma_{\mathrm{PT}}(f)=\gamma_{\mathrm{PT}}(m)$.
\end{assumption}
\noindent
Assumption \ref{A.td} requires that trend differences are equal in levels between genders. Assumption \ref{A.ntd_pt} can hence be viewed as a percentage analog of Assumption \ref{A.td}. Given Proposition~\ref{prop:cp_ident_ass}, Assumption~\ref{A.ntd_pt} can be interpreted as the implicit identification assumption underlying the normalized event-study empirical strategy (Section~\ref{sec:event_studies}). Going forward, I refer to the identification framework based on Assumption~\ref{A.ntd_pt} as the normalized triple differences (NTD) framework.

\subsection{Non-Identification of the Conventional Target Estimand}\label{sec:bias_char}
The gender gap in normalized DID is what normalized event studies aim to estimate. Its target causal estimand is the gender gap in normalized effects. The following result characterizes how this descriptive estimand is biased for its target causal estimand under NTD.

\begin{theorem}\label{thm:cp_iden_res}
Consider a $2\times2$ with treatment group $d$, post-treatment age $a$, and control group $d^\prime$ such that $d^\prime >a\geq d$, and assume Assumptions \ref{A.no.ant.1} and \ref{A.ntd_pt} hold. Assume further that $\APO(g,\infty)-\gamma_{\mathrm{PT}}(g)\neq0$ for $g\in\{f,m\}$. Then
\begin{equation*}
\delta_{\theta}(f) - \delta_{\theta}(m)
= \mathrm{Bias} \times \big[ \theta(f) - \theta(m) \big],
\end{equation*}
where
\begin{equation*}
\mathrm{Bias} = \frac{\APO(f,\infty)}{\APO(f,\infty) - \gamma_{\mathrm{PT}}(f)} = \frac{\APO(m,\infty)}{\APO(m,\infty) - \gamma_{\mathrm{PT}}(m)}.
\end{equation*}
\end{theorem}
\noindent
Theorem \ref{thm:cp_iden_res} shows the descriptive gender gap in normalized DID does not identify the gender gap in normalized effects under NTD when the parallel trends assumption (Assumption \ref{A.pt}) is violated. The bias term is gender invariant under Assumption~\ref{A.ntd_pt}.\footnote{Lemma \ref{lem:ntd_cross_bias} in Appendix \ref{sec:app_ident_proofs} shows that $P=[\ATE(f)-\ATE(m)]/\APO(f,\infty)$, the $2\times2$ analogue of the estimator that \citet{kleven2019children} term the child penalty ($P_t$ in their notation), is also not identifiable under NTD when parallel trends in levels fail.} 
The mechanism of the bias in Theorem \ref{thm:cp_iden_res} is made transparent in the following lemma, which shows how each descriptive estimand relates to its causal counterpart.

\begin{lemma}\label{lemma:desc_tie_cause}
Consider a $2\times2$ with treatment group $d$, control group $d^\prime>a$, pre-treatment age $d-1$, and target age $a$, which can be pre- or post-treatment. Assume Assumption~\ref{A.no.ant.1} holds for both $g\in\{f,m\}$ and treatment groups $\{d,d^\prime\}$, and $\APO(g,\infty)-\gamma_{\mathrm{PT}}(g)\neq0$. Then
\begin{align}
\delta_{\mathrm{APO}}(g) & = \APO(g,\infty) - \gamma_{\mathrm{PT}}(g), \label{eq:app_proof_cp_ass_1}\\
\delta_{\mathrm{ATE}}(g) & = \ATE(g) + \gamma_{\mathrm{PT}}(g), \label{eq:app_proof_cp_ass_2}\\
\delta_{\theta}(g) & = \theta(g) \frac{\APO(g,\infty)}{\APO(g,\infty) - \gamma_{\mathrm{PT}}(g)} + \frac{\gamma_{\mathrm{PT}}(g)}{\APO(g,\infty) - \gamma_{\mathrm{PT}}(g)}. \label{eq:lemma_d_theta_bias}
\end{align}
\end{lemma}
\noindent
Lemma \ref{lemma:desc_tie_cause} shows that the within-gender normalized DID \eqref{eq:lemma_d_theta_bias} is composed of two terms. The second term is the parallel-trends violation relative to the imputed counterfactual; under NTD it is equal across genders and cancels when differencing. The first term is the causal normalized effect $\theta(g)$ multiplied by the ratio of true to imputed counterfactual earnings, the $\mathrm{Bias}$ term in Theorem \ref{thm:cp_iden_res}. Under NTD this ratio is also equal across genders, but it is multiplicative rather than additive, so it does not cancel when differencing.\footnote{Since the bias is multiplicative, if the true gender gap is zero then the descriptive gap is zero as well.}

To summarize, NTD does not allow identification of the gender gap in normalized effects in a single $2\times2$. Since the event-study estimator aggregates many such $2\times2$s, it inherits their biases and does not identify its target causal estimand when parallel trends is violated. 

\subsection{A Bias-Bounding Exercise}\label{sec:bias_bound}
Theorem \ref{thm:cp_iden_res} shows the conventional target is not identified when parallel trends are violated, and Section \ref{sec:did_violation} argues that they are violated, mainly for early treatment groups (20--25) when compared to later treatment groups (30+). A question remains how large the resulting biases are.
I now develop a bias-bounding exercise for the conventional estimator, building on a bias-correction approach that assumes fathers' effects are known. Under a plausible range of fathers' effects, empirical results suggest substantial bias for earlier treatment groups, 24--51\%, with the conventional estimator understating child penalties. For later treatment groups, the exercise provides no evidence of bias.

\paragraph{Theory.}
The following result formalizes a bias-correction approach to identify the conventional causal estimand, the gender gap in normalized effects.

\begin{proposition}\label{prop:bias_correction}
Under the setup of Theorem \ref{thm:cp_iden_res}, if $\APO(m,\infty)$ is known, then
\begin{equation*}
    \theta(f)-\theta(m)=\big[\delta_{\theta}(f)-\delta_{\theta}(m)\big]\frac{\delta_{\mathrm{APO}}(m)}{\APO(m,\infty)}.
\end{equation*}
\end{proposition}

\noindent
For a sketch of the proof, note that if fathers' counterfactual APO is known, then the parallel trends violation for fathers is identified, which in turn identifies the $\mathrm{Bias}$ term in Theorem \ref{thm:cp_iden_res}. The result follows by multiplying the expression in Theorem \ref{thm:cp_iden_res} by the inverse of the bias term. Note that since observed earnings identify the APO under realized treatment ($\APO(g,d)$), an assumption on $\APO(m,\infty)$ is equivalent to an assumption on fathers' normalized effects ($\theta(m)$), and vice-versa.

An immediate limitation is that fathers' counterfactual APO is not known. I therefore use the bias-correction approach as a diagnostic tool. Since both the conventional and bias-corrected estimators target the same causal estimand, assuming a plausible range for fathers' effects allows us to bound the bias in the conventional estimator. Specifically, if the true effect lies within an assumed range, comparing the conventional estimator to the bias-corrected estimator provides a measure of bias magnitude. Below, I empirically implement this exercise using a range from $-10\%$ to $+10\%$ for fathers' normalized effects. To my knowledge, there are two quasi-experimental studies that exploit the random success of in vitro fertilization (IVF) treatments and study fathers, and find no significant earnings effect for men \citep{lundborg2024there} or, if anything, a small positive effect \citep{bensnes2023reconciling}.\footnote{An alternative is to bound the parallel-trends violation using pre-treatment violations \citep{rambachan2023more}. The human capital mechanism of Section~\ref{sec:did_violation} suggests against this approach in the considered context: at treatment-relevant ages the later-treated control group transitions from schooling to employment, so violations can change sign and grow at treatment onset, making pre-treatment violations an unreliable benchmark for post-treatment ones. I therefore do not bound fathers' effects using pre-trends diagnostics.}

\paragraph{Empirical Application.}
Figure~\ref{fig:ntd_robustness} implements the bias-bounding exercise separately by treatment group. The black series (``No Assumption'') shows the conventional estimator: the gender gap in normalized DID. The shaded band shows the range of bias-corrected estimates from Proposition~\ref{prop:bias_correction} as $\theta(m)$ varies over $[-0.1, 0.1]$. Estimators are constructed using sample analogs of population expectations, and standard errors are constructed using influence functions and clustered at the individual level. See Appendix~\ref{sec:appendix_no_covariates} for further discussion on estimators and standard errors.

\begin{figure}[ht!]
    \centering
    \includegraphics[width=\textwidth]{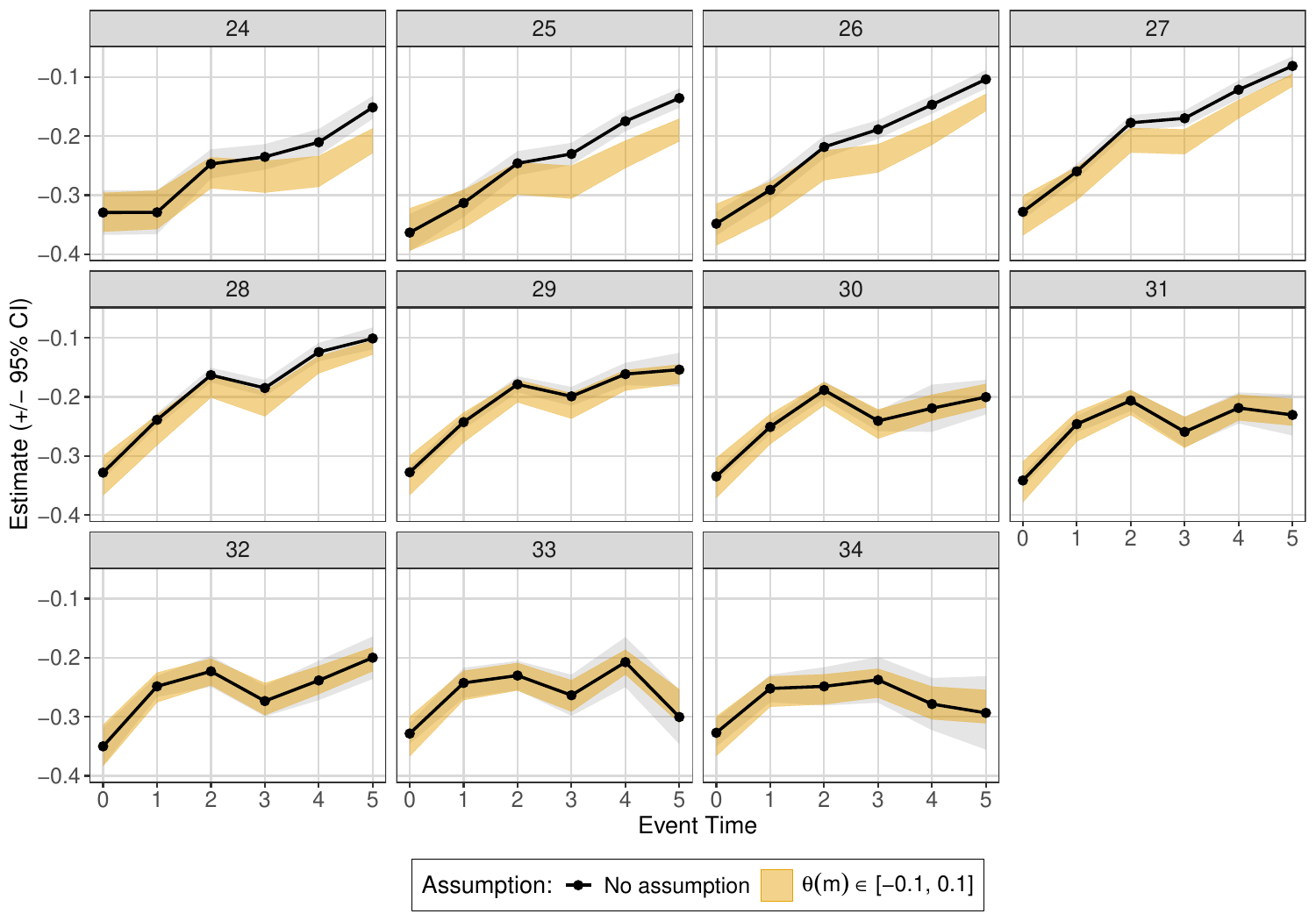}
    \caption{Bias-Bounding Diagnostic of the Conventional Estimator. \textit{Notes:} The figure presents the bias-bounding exercise by treatment group (facets). The black line (``No Assumption'') shows the conventional estimator, the gender gap in normalized DID, without imposing any value for $\theta(m)$, with 95\% confidence intervals in gray. The shaded band shows the range of bias-corrected estimates as $\theta(m)$ varies over $[-0.1, 0.1]$, building on the discussion in Section \ref{sec:bias_bound}.}
    \label{fig:ntd_robustness}
\end{figure}

The results suggest the bias is heterogeneous across treatment groups. At five years post-treatment, the conventional estimator for earlier treatment groups ($D=24,25,26$) is less negative than bias-corrected estimates throughout the assumed range of $\theta(m)$. To illustrate, consider treatment group $D=26$: the conventional estimate is $-0.104$ (SE $0.008$). Assuming $\theta(m)=-0.10$ yields a bias-corrected estimate of $-0.129$, while $\theta(m)=0.10$ yields $-0.157$—that is, $24\%$ and $51\%$ larger in magnitude, respectively. In contrast, for later treatment groups ($D\geq28$), conventional and bias-corrected estimates are similar, providing no evidence of bias.\footnote{The assumed range also carries an implication for the male DID itself. For early treatment groups the male normalized DID is large: $-0.27$ for $D=26$ five years post-treatment. Under no parallel-trends violation, this implies fathers lose 27\% of counterfactual earnings due to parenthood, far outside the quasi-experimental evidence. In contrast, assuming fathers' true effect is $-10\%$, an observed effect of $-0.27$ implies a normalized parallel trends violation of $-24\%$. }
At face value, the results indicate that for earlier treatment groups the conventional estimator attenuates child penalties toward zero. This is consistent with Section~\ref{sec:did_violation}, which argued parallel trend violations are negative ($\gamma_{\mathrm{PT}}(g)<0$) when comparing early to late treatment groups. Under NTD, $\gamma_{\mathrm{PT}}(g)<0$ causes child penalties to be understated (Theorem~\ref{thm:cp_iden_res}).

\section{Identification of the Effect of Parenthood on the Gender Earnings Ratio}\label{sec:alt_estimand}
The above discussion established that under NTD the gender gap in normalized effects is not identified when parallel trends fail. This section establishes a new identification result: the effect of parenthood on the gender earnings ratio is identified under NTD. Estimates of this new estimand show a novel empirical finding: the extent to which parenthood accounts for gender inequality in labor market earnings diminishes for older treatment groups. I conclude by comparing the new and conventional estimators.

\subsection{Identification and Estimation}\label{sec:iden_ntd_new}
Let $\rho(d,d^\prime,a)=\tfrac{\APO(f,d,d^\prime,a)}{\APO(m,d,d^\prime,a)}$ denote the gender earnings ratio for treatment group $d$ at age $a$ under counterfactual treatment timing $d^\prime$. In suppressed notation for a single $(d,d^\prime,a)$ I write $\rho(d,d^\prime,a)\equiv\rho(d^\prime)$, so that $\rho(d)\equiv\rho(d,d,a)$ is the realized ratio and $\rho(\infty)\equiv\rho(d,\infty,a)$ the counterfactual ratio absent childbirth. The next result shows the effect of parenthood on the gender earnings ratio, $\rho(d)-\rho(\infty)$, is identified under NTD.

\begin{theorem}\label{thm:ntd_new_estimand}
Assume the same setup as in Theorem~\ref{thm:cp_iden_res}. Then
\[
\rho(d)-\rho(\infty)=\frac{\mathbb{E}[Y_a\mid G=f,D=d]}{\mathbb{E}[Y_a\mid G=m,D=d]}
-\frac{\delta_{\mathrm{APO}}(f)}{\delta_{\mathrm{APO}}(m)}.\]
\end{theorem}

\noindent
The proof is provided in Appendix~\ref{sec:app_ident_proofs}. A sketch of the proof is as follows. The gender earnings ratio under the realized treatment $\rho(d)$ is identified directly from observed data by consistency, while the counterfactual ratio $\rho(\infty)$ is identified via the gender ratio of $\delta_{\mathrm{APO}}$ under Assumption~\ref{A.ntd_pt}.\footnote{The effect on the gender earnings ratio is also identified under assumptions that identify each gender's counterfactual mean earnings separately. For example, if Assumption~\ref{A.pt} holds within each gender, then each $\APO(g,\infty)$, and hence $\rho(\infty)$, is identified without adopting NTD.}

Estimation follows directly from Theorem~\ref{thm:ntd_new_estimand}: replace population expectations with sample means to obtain estimates for the observed earnings ratio and the $\delta_{\mathrm{APO}}$ gender ratio, then difference the two ratios to construct the final estimator. For inference, I derive clustered standard errors based on the estimator's influence function in Appendix~\ref{sec:appendix_no_covariates}.

\subsection{Empirical Evidence}\label{sec:new_ntd_empirical}

\paragraph{Validation Tests.}
Section~\ref{sec:violations_validation_did} presented validation tests for DID. Building on that discussion, Figure~\ref{fig:validation_ntd} reports pre-treatment estimates of $\delta_{\theta}(f)-\delta_{\theta}(m)$ by treatment-control pairs, which serves as a validation test for Assumption~\ref{A.ntd_pt} (Proposition~\ref{prop:cp_ident_ass}). For mid-range treatment groups $D=26,\ldots,30$ the normalized pre-trends are small and, unlike the DID pre-trends, do not widen with the treatment-control treatment timing gap. I therefore focus on these treatment groups when estimating the new estimand below.

\begin{figure}[t!]
    \centering
    \includegraphics[width=\textwidth]{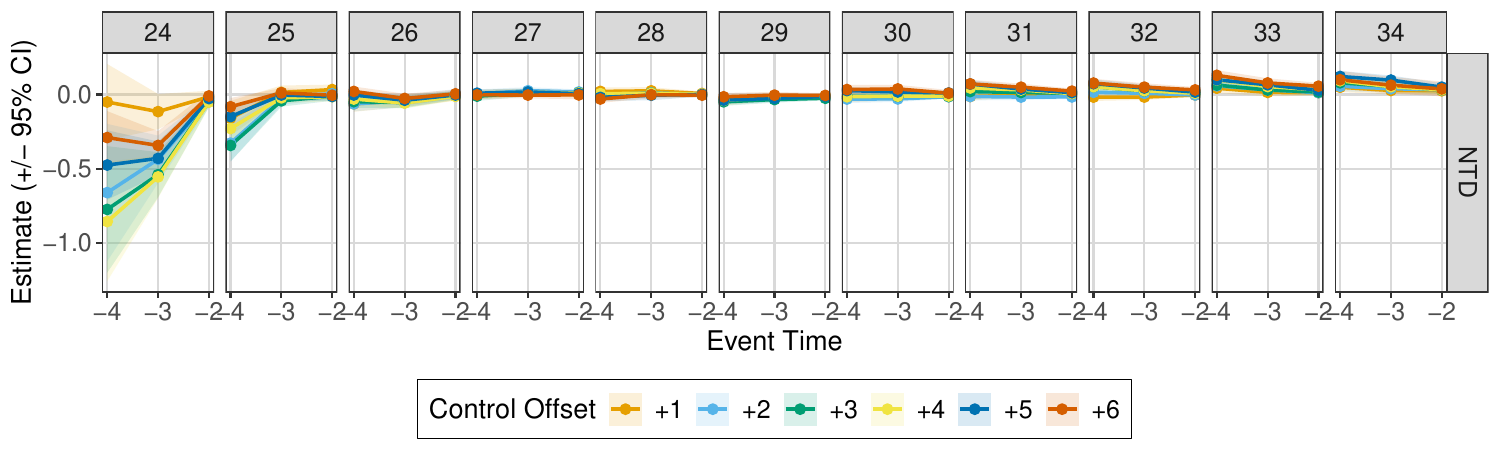}
    \caption{NTD Validation Tests by Treatment-Control Pairs. \textit{Notes:} The figure presents pre-treatment estimates of the gender gap in normalized DID, $\delta_{\theta}(f)-\delta_{\theta}(m)$. For further notes see Figure \ref{fig:validation_pretrends}.}
    \label{fig:validation_ntd}
\end{figure}

\paragraph{Estimating the Effect of Parenthood on Gender Inequality.}
The new estimand can be visualized as the difference between two gender earnings ratios. Figure~\ref{fig:alt_ntd_a} reports post-treatment estimates of the counterfactual gender earnings ratio $\frac{\delta_{\mathrm{APO}}(f)}{\delta_{\mathrm{APO}}(m)}$ in purple and the observed gender earnings ratio $\frac{\mathbb{E}[Y_a\mid G=f,D=d]}{\mathbb{E}[Y_a\mid G=m,D=d]}$ in green, separately by treatment group. Under NTD, the purple series identifies the gender earnings ratio in the counterfactual absent childbirth ($\rho(\infty)$), and the green series identifies the ratio under realized treatment ($\rho(d)$). The difference between these series estimates the effect of parenthood on the gender earnings ratio, reported in Figure~\ref{fig:alt_ntd_b} via the orange series.

\begin{figure}[!ph]
    \centering
    \begin{subfigure}[b]{\textwidth}
        \centering
        \caption{Components of the New Estimator}
        \label{fig:alt_ntd_a}
        \includegraphics[width=\textwidth]{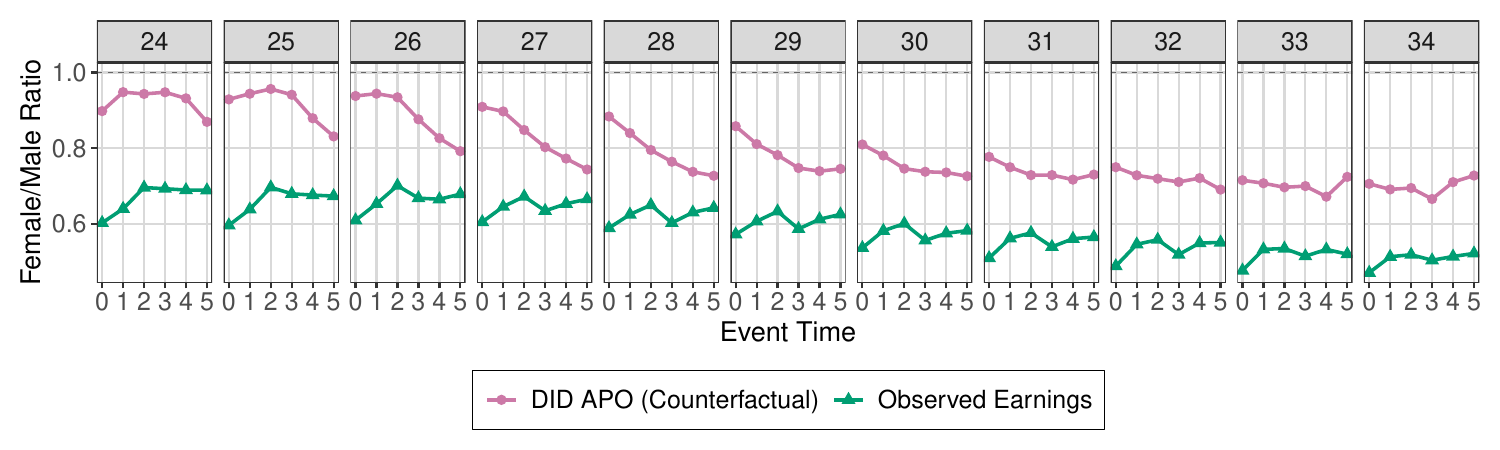}
    \end{subfigure}

    \vspace{1em}

    \begin{subfigure}[b]{\textwidth}
        \centering
        \caption{Estimators of the Effect of Parenthood}
        \label{fig:alt_ntd_b}
        \includegraphics[width=\textwidth]{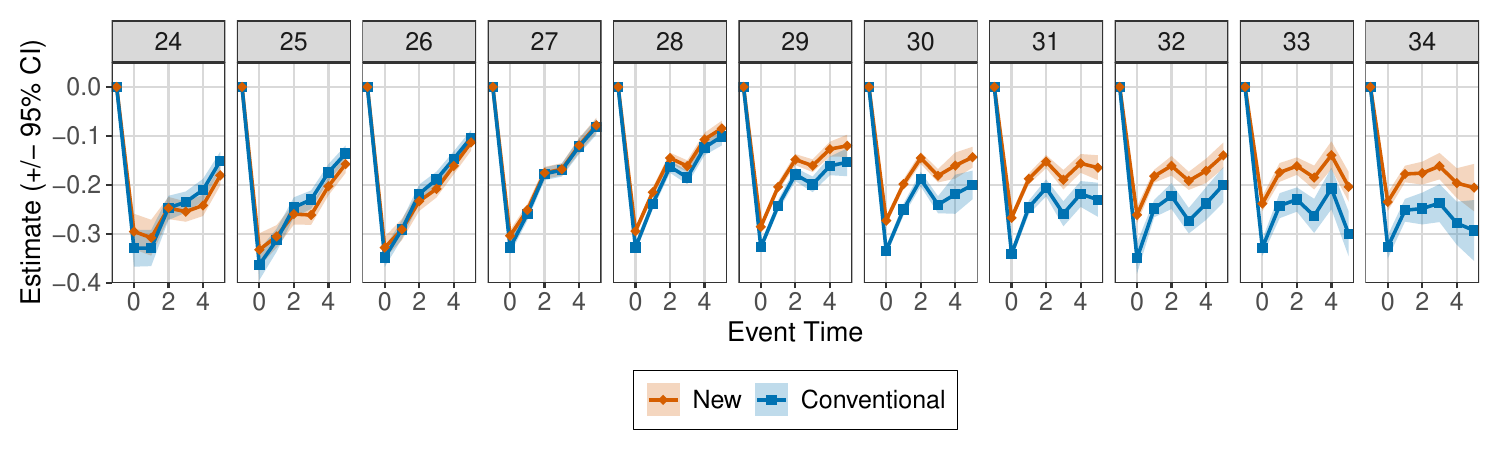}
    \end{subfigure}

    \vspace{1em}

    \begin{subfigure}[b]{\textwidth}
        \centering
        \caption{Accounting for the Gender Gap}
        \label{fig:alt_ntd_c}
        \includegraphics[width=\textwidth]{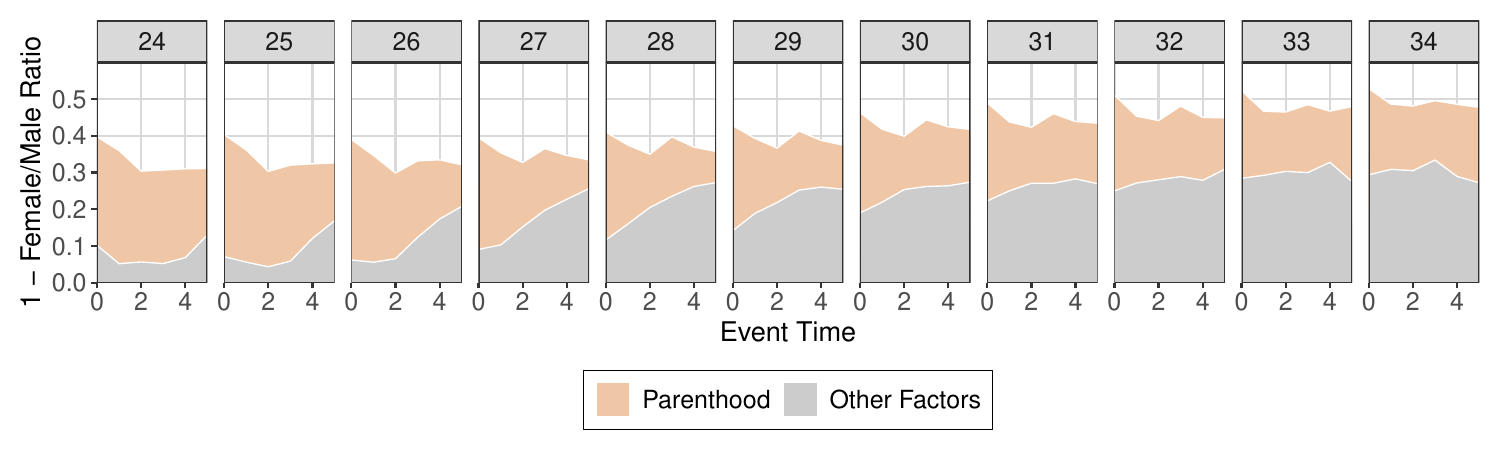}
    \end{subfigure}

    \caption{Estimates of the Effect of Parenthood on the Gender Earnings Ratio. \textit{Notes:} The x-axis represents event time: years since first childbirth. Panel~(a) presents the DID-imputed counterfactual gender earnings ratio (purple circles) and the observed gender earnings ratio (green triangles). Their difference estimates the effect of parenthood on the gender earnings ratio, reported in panel~(b) (orange diamonds) alongside the conventional estimator (blue squares), the gender gap in normalized DID. Panel~(c) splits the observed gender gap, one minus the observed ratio, into the part due to parenthood (orange area), the counterfactual ratio minus the observed ratio, and the part due to other factors (gray area), one minus the counterfactual ratio.}
    \label{fig:alt_ntd}
\end{figure}

Figure~\ref{fig:alt_ntd_a} shows that following childbirth, the gender ratio of observed earnings falls sharply relative to the counterfactual ratio in all treatment groups. The gap, which corresponds to the estimator of the effect of parenthood on the gender earnings ratio and is reported in Figure~\ref{fig:alt_ntd_b}, persists throughout the post-treatment horizon. For example, for treatment group $D=30$, parenthood reduces the gender earnings ratio by $14.4$ percentage points (SE $1.1$) five years post-childbirth.

\paragraph{Accounting for Parenthood in Gender Inequality.} An important accounting exercise in the child penalty literature is to quantify how much of the overall gender gap is attributable to parenthood \citep[see, e.g.,][]{kleven2019children,cortes2023children}. This can be achieved directly using the values reported in Figure~\ref{fig:alt_ntd_a}. To illustrate, consider treatment group $D=30$ at the age of first childbirth, i.e., age 30. The estimated counterfactual gender earnings ratio is $0.810$ and the realized gender earnings ratio is $0.537$. This implies an overall gender inequality of $1 - 0.537 = 0.463$. Of this, parenthood accounts for $0.810 - 0.537 = 0.273$, or approximately $59\%$, while other factors account for the remaining $41\%$. Figure~\ref{fig:alt_ntd_c} reports this decomposition for all treatment groups and event times, and reveals substantial heterogeneity: at the age of first childbirth, parenthood's share ranges from 84\% for $D=26$ down to 59\% for $D=30$. Within each treatment group, parenthood's share also declines with event time: for $D=26,\ldots,30$, from between 59\% and 84\% at childbirth to between 23\% and 35\% five years later.

\subsection{Comparing to the Conventional Estimator}\label{sec:ntd_new_discussion}

The relationship between the new causal estimand, the effect of parenthood on the gender earnings ratio, and the conventional causal estimand, the gender gap in normalized effects, can be expressed as:
\begin{equation}\label{eq:ntd_base_to_alt}
\rho(d)-\rho(\infty)
= \frac{\rho(\infty)}{1 + \theta(m)}
\big[\theta(f) - \theta(m)\big].
\end{equation}
Equation~\eqref{eq:ntd_base_to_alt} expresses the new estimand as the conventional estimand scaled by the factor $\rho(\infty)/(1+\theta(m))$. The two coincide when this factor equals one, and the new estimand is smaller in absolute value whenever the factor is below one.
For example, if women and men earn the same absent childbirth and childbirth reduces women's earnings by 50\% while leaving men's unchanged, the factor equals one and both estimands equal $-50$ percentage points (pp): $-50\%-0\%$ for the conventional and $100\%\times0.5-100\%$ for the new.\footnote{To compute the new estimand by hand: let $M$ and $F$ be men's and women's earnings absent childbirth, so $\rho(\infty)=F/M$. Childbirth scales these to $M(1+\theta(m))$ and $F(1+\theta(f))$, so the realized ratio is $\rho(d)=\frac{F(1+\theta(f))}{M(1+\theta(m))}=\rho(\infty)\frac{1+\theta(f)}{1+\theta(m)}$.}
If instead women earn only 80\% of men absent childbirth, with men's earnings still unaffected, the factor is smaller than one: the conventional estimand is again $-50$ pp, while the new is $80\%\times0.5-80\%=-40$ pp.
Similarly, if instead women and men earn the same absent childbirth but childbirth raises men's earnings by 20\%, the factor is again smaller than one, and the conventional estimand is $-50\%-20\%=-70$ pp, while the new is $100\%\times0.5/1.2-100\%\approx-58$ pp.

To empirically study how counterfactual gender inequality $\rho(\infty)$ evolves across the life cycle and treatment groups, I compute mean earnings at every age $a<d$ for each treatment group $d$ and form the female-to-male observed earnings ratio. Under the no-anticipation assumption (Assumption~\ref{A.no.ant.1}), these pre-childbirth earnings identify the counterfactual $\APO$s.
Figure~\ref{fig:rho_lifecycle} plots the female-to-male observed earnings ratio for each treatment group $d\in[27,38]$. Two patterns emerge. First, within each treatment group the gender earnings ratio before childbirth follows a similar life-cycle pattern: rising up to age 27 and falling from age 28, producing an inverted U-shape. Second, at any given age, parents who delay their first birth generally exhibit higher ratios than earlier-childbearing parents. Importantly, the life-cycle pattern is such that at later ages all groups have a gender earnings ratio smaller than one. Under no anticipation, this implies $\rho(d,\infty,a) < 1$ at these ages.

\begin{figure}[t!]
    \centering
    \includegraphics[width=\textwidth]{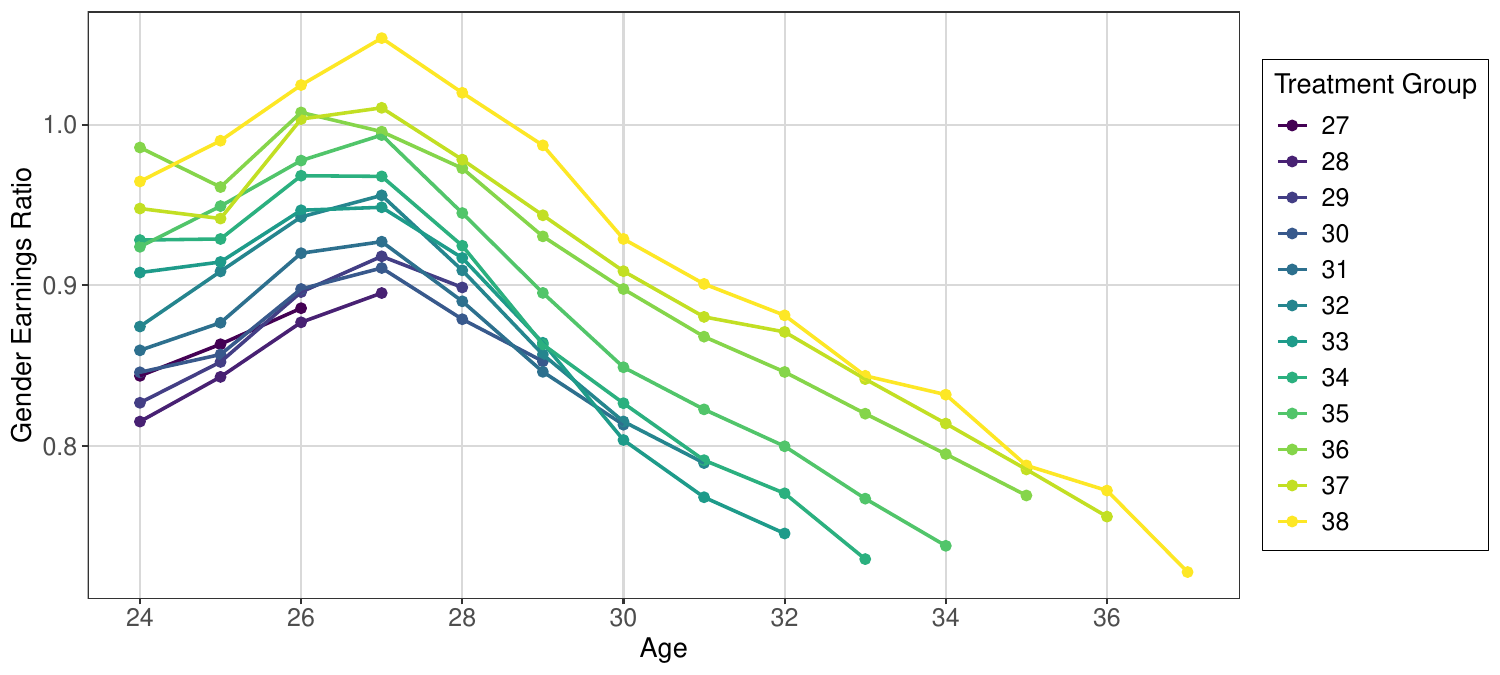}
    \caption{Observed Gender Earnings Ratios Before Childbirth. \textit{Notes:} The figure presents the ratio of mean female to mean male earnings (y-axis) for pre-treatment ages separately by treatment group (colors). The sample is restricted to treatment groups $D\in[27,38]$. The dashed horizontal line indicates gender equality in earnings.}
    \label{fig:rho_lifecycle}
\end{figure}

Turning to the comparison, Figure~\ref{fig:alt_ntd_b} plots both the conventional and new estimators.
For $D=26,27$, the new estimator starts slightly less negative than the conventional at childbirth but converges and becomes marginally more negative by five years post-treatment. For $D=29,30$, the new estimator is smaller in magnitude than the conventional.
The bias-bounding exercise in Section~\ref{sec:bias_bound} found no evidence of substantial bias for later treatment groups. Hence, the differences observed here likely reflect different target estimands. Furthermore, the smaller magnitude of the new estimator is consistent with small effects for fathers and $\rho(\infty) < 1$, following \eqref{eq:ntd_base_to_alt}.

\section{Aggregation}\label{sec:aggregation}

Applied work typically reports estimates aggregated across all treatment groups, as in the normalized event studies (Section~\ref{sec:event_studies}). This section discusses how to define aggregate estimands in the context of child penalties, building on the discussion of aggregation in \citet{callaway2021difference}, and highlights challenges in comparing aggregates across strata, such as countries or parent types, due to differences in treatment distributions.

\subsection{Aggregate Causal Estimands}\label{sec:agg_estimands}

A natural aggregate by event time $e$ across treatment groups of the effect of parenthood on the gender earnings ratio is
\begin{equation}\label{eq:agg_rho}
\rho_{\mathrm{Agg}}(e)=\mathbb{E}_D\big[ \rho(D,D,D+e) - \rho(D,\infty,D+e) \mid D+e<D_{\max}\big].
\end{equation}
$\rho_{\mathrm{Agg}}(e)$ is a weighted average of treatment-group-specific effects on the gender earnings ratio at event time $e$, using the treatment distribution as weights, where $D_{\max}$ is the largest treatment group considered.\footnote{Comparing values across $e$ combines variation in causal effects with shifts in treatment group composition, as noted by \citet{callaway2021difference}. This can be addressed by conditioning on a maximum exposure length $e'$, and comparing $e_1, e_2 \leq e'$.}

For the conventional normalized effect $\theta(g)$, aggregation introduces additional subtleties. Within gender $g\in\{f,m\}$, a natural aggregate across treatment groups by event time is
\begin{align*}
\theta_{\mathrm{Agg},1}(g, e) &= \mathbb{E}_D\left[\theta(g, D, D+e) \mid G = g, D+e < D_{\max}\right].
\end{align*}
$\theta_{\mathrm{Agg},1}(g, e)$ is a weighted average of normalized effects for gender $g$ at event time $e$.
However, the event study estimator $\widehat{\theta}_{\mathrm{ES}}(g, e)$ in Equation~\eqref{eq:cp_main_res} is a ratio of two averages. Hence, an aggregate estimand that matches the structure of $\widehat{\theta}_{\mathrm{ES}}(g, e)$ is
\begin{equation*}
\theta_{\mathrm{Agg},2}(g, e) = \frac{\mathbb{E}_D\left[\ATE(g, D, D+e) \mid G = g, D+e < D_{\max}\right]}{\mathbb{E}_D\left[\APO(g, D, \infty, D+e) \mid G = g, D+e < D_{\max}\right]}.
\end{equation*}
Comparing, while $\theta_{\mathrm{Agg},1}(g,e)$ uses the treatment distribution to summarize effects, $\theta_{\mathrm{Agg},2}(g, e)$ effectively gives higher weight to treatment groups with higher counterfactual APOs, which are typically the later-treated, as their post-treatment periods occur in later parts of the life cycle compared to earlier treatment groups. To see this, let $p(g,d,e)=\Pr(D=d\mid G=g,\ D+e<D_{\max})$. Substituting $\ATE(g)=\theta(g) \APO(g,\infty)$ into $\theta_{\mathrm{Agg},2}(g, e)$ and simplifying, we get
\[
\theta_{\mathrm{Agg},2}(g,e)=\sum_{D:D+e<D_{\max}}w(g,D,e)\theta(g,D,D+e),
\]
where 
\[
w(g,d,e)=\frac{p(g,d,e)\APO(g,d,\infty,d+e)}{\sum_{d^{\prime}:d^{\prime}+e<D_{\max}}p(g,d^{\prime},e)\APO(g,d^{\prime},\infty,d^{\prime}+e)}.
\]
Hence $\theta_{\mathrm{Agg},2}$ is also a weighted average of normalized effects for gender $g$ at event time $e$, like $\theta_{\mathrm{Agg},1}(g, e)$, with higher weights to groups with higher counterfactual APOs. Since giving higher weights to higher earning treatment groups does not seem warranted, I see $\theta_{\mathrm{Agg},1}(g,e)$ as the preferable aggregate estimand.

A further question is on the order of aggregation and differencing: whether to aggregate within gender first and then take the difference, $\theta_{\mathrm{Agg},1}(f, e) - \theta_{\mathrm{Agg},1}(m, e)$, or take the within-group difference and aggregate across groups,
\begin{equation}\label{eq:agg_conv}
{\Delta\theta_{\mathrm{Agg}}}(e) = \mathbb{E}_D\big[\theta(f, D, D+e) - \theta(m, D, D+e) \mid D+e < D_{\max}\big].
\end{equation}
The two coincide only when male and female treatment distributions are identical. Below, I use $\Delta\theta_{\mathrm{Agg}}(e)$ for two reasons. First, to ease comparison to $\rho_{\mathrm{Agg}}(e)$ from \eqref{eq:agg_rho}, which uses the same weights as $\Delta\theta_{\mathrm{Agg}}(e)$. Second, $\theta_{\mathrm{Agg},1}(f, e) - \theta_{\mathrm{Agg},1}(m, e)$ confounds differences in treatment-group-specific effects with gender differences in treatment distribution.

Aggregate estimates for Israel, the UK, and Germany are reported in Appendix~\ref{sec:agg_estimates_app}, and are consistent with the single-treatment-group patterns of Section~\ref{sec:new_ntd_empirical}.

\subsection{Comparing Aggregate Estimates Across Strata}\label{sec:agg_crosscountry}
In practice, studies calculate aggregate estimates within a stratum, e.g., countries \citep[e.g.,][]{kleven2019child} or parent type \citep[e.g.,][]{andresen2022causes}, and compare across strata. The caveat in interpretation with such comparison, putting aside the obvious reasons of differing data samples, is that the difference may be affected by differences in treatment distributions. For example, countries where parents have children earlier place more weight on early treatment groups, potentially yielding different aggregate estimates than countries with later childbearing. 

Figure~\ref{fig:oecd_age_dist} documents cross-country variation in the distribution of age at first childbirth for six OECD countries: the United States and Poland exhibit a right-skewed distribution of first-birth ages; Denmark and Sweden are more centered; Italy and Spain are left-skewed.
Accordingly, the United States may display a different aggregate than Italy simply because it places greater weight on earlier treatment groups.

\begin{figure}[ht!]
    \centering
    \includegraphics[width=\textwidth]{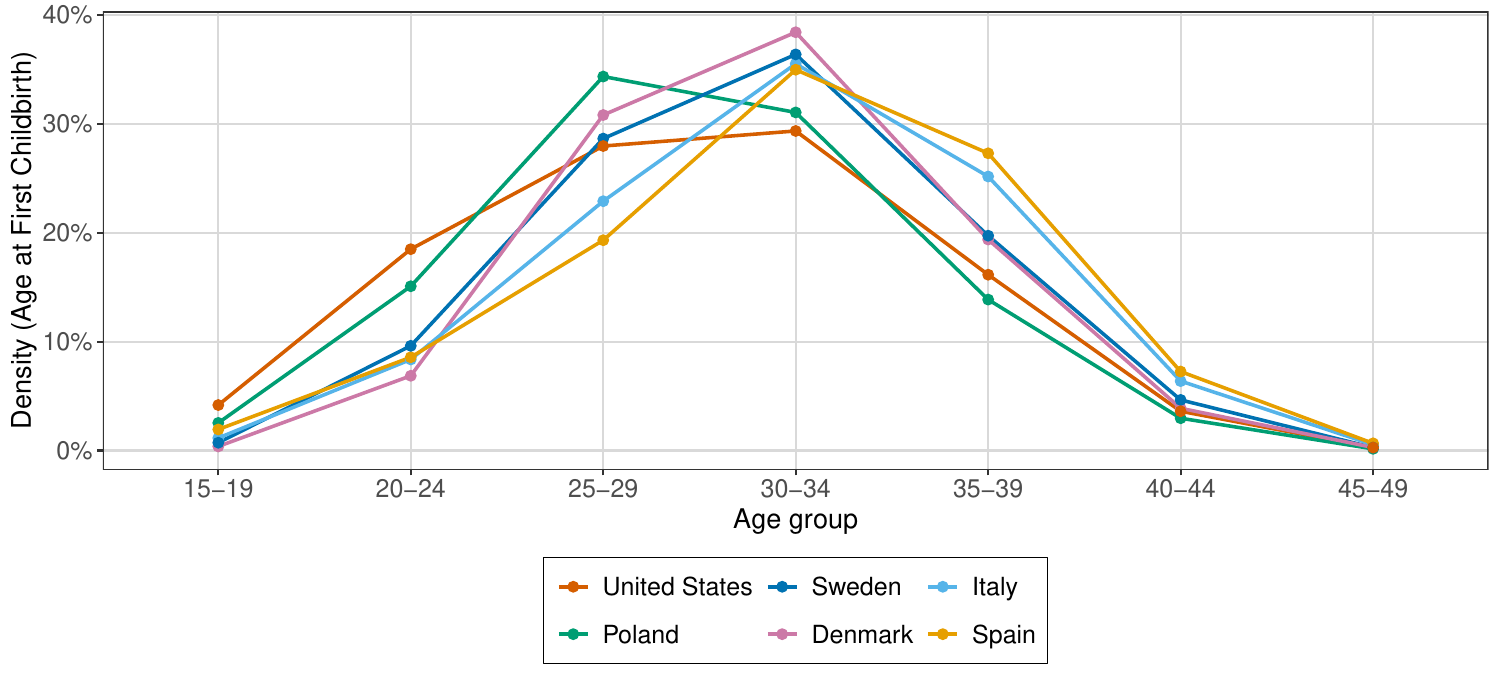}
    \caption{OECD Age at First Childbirth Distributions. \textit{Source:} OECD Family Database (Indicator SF2.3). \textit{Notes:} The figure shows the distribution of age at first childbirth for selected OECD countries in 2021. The OECD database reports fertility rates for first births by five-year age groups (15--19, 20--24, \ldots, 45--49). For each country, the age-group fertility rates are normalized to sum to one, so the vertical axis represents the within-country share of first births across age groups.}
    \label{fig:oecd_age_dist}
\end{figure}

To illustrate how differences in treatment distributions may affect results, Figure~\ref{fig:example_aggregates} reports aggregated estimates of the effect of parenthood on the gender earnings ratio, i.e., the estimated $\rho_{\mathrm{Agg}}(e)$, keeping the single-treatment-group estimates constant using the estimates of Section \ref{sec:alt_estimand}, and varying the treatment distribution. The three considered treatment distributions are motivated by the previous discussion on OECD distributions: one gives more weight to earlier treatment groups (right-skewed), the second gives more weight to mid-range treatment groups (centered), and the third places more weight on later treatment groups (left-skewed).

\begin{figure}[ht!]
    \centering
    \includegraphics[width=\textwidth]{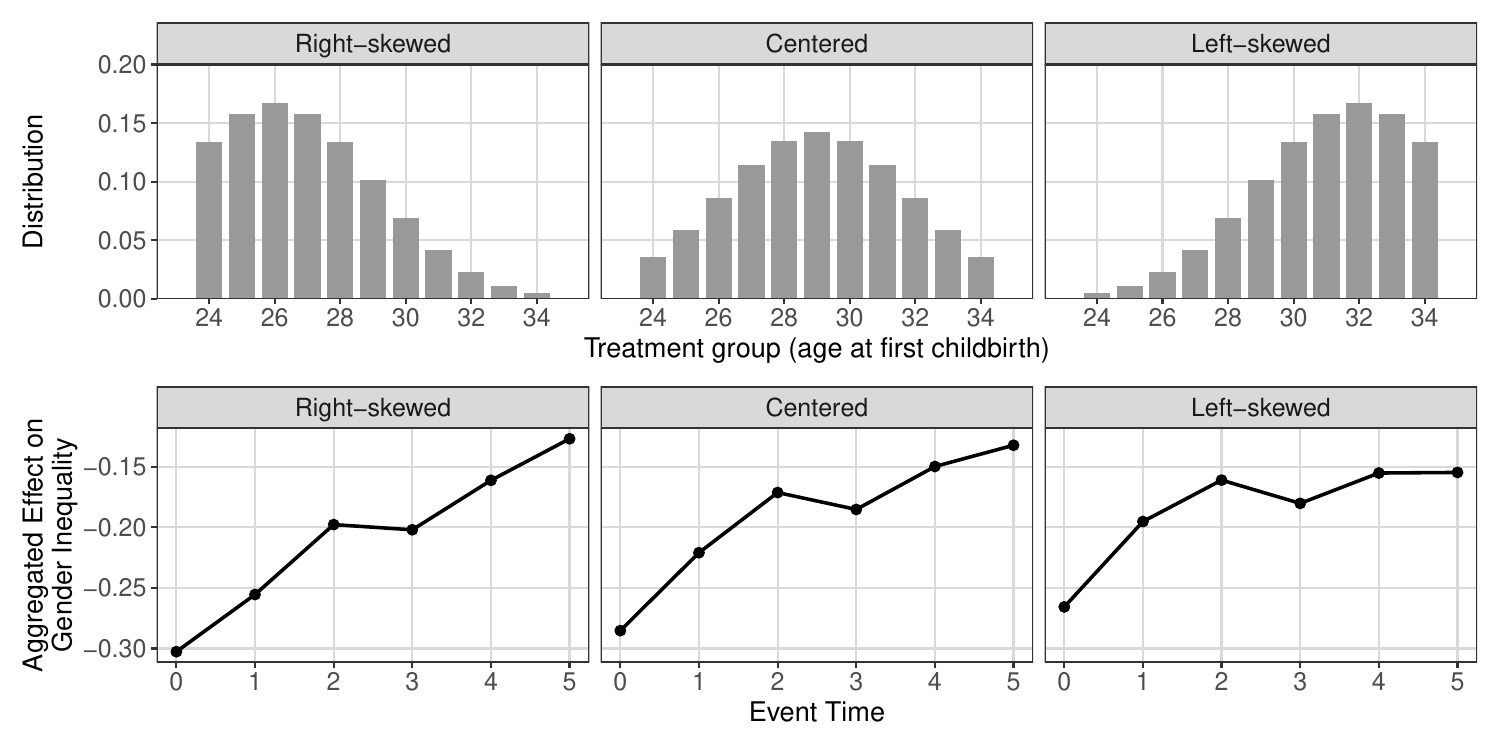}
    \caption{Aggregated Effects Under Three Hypothetical Treatment Distributions. \textit{Notes:} The figure shows estimates of the aggregated effect of parenthood on the gender earnings ratio across treatment groups $D\in[24,34]$ by treatment distribution (columns). Single-treatment-group estimates were calculated on the Israeli administrative data, as in Section~\ref{sec:alt_estimand}. The three treatment distributions are based on normal distributions with standard deviation 3 and means 27, 29 and 31, titled Right-skewed, Centered and Left-skewed, respectively. Aggregated estimates report weighted averages of single-treatment-group estimates using the density of each treatment group as weights.}
    \label{fig:example_aggregates}
\end{figure}

The documented aggregated estimates show two different patterns. Under the right-skewed distribution, the effect of parenthood on gender inequality in earnings becomes smaller over time almost linearly. Fitting a line to these estimates, the effect would reach zero around event time 10. In contrast, under the left-skewed distribution the aggregate drops from 27\% to 16\% between event times zero and two, but then stays around 15\% up to event time five. That is, under this distribution, one can conjecture that the effect of parenthood on gender inequality does not recede over time. This exercise illustrates that when single-treatment-group effects are heterogeneous, differences in treatment distributions can substantially affect aggregated estimators.

\section{Conclusion}\label{sec:conclusion}

This paper revisits the identification of child penalties in the normalized event-study framework. 
Within gender, I argue that difference-in-differences estimates are biased by selection into the timing of parenthood.
Between genders, I articulate the identification assumptions underlying the normalized design, derived from the validation test used in applied work: parallel-trend violations, divided by counterfactual earnings, are equal across genders, which I term normalized triple differences (NTD). 
I show that under NTD, the conventional target, the gender gap in normalized effects, is not identified when parallel trends is violated, and that the effect of parenthood on the gender earnings ratio is point identified.
Applied to Israeli data, the new estimator reveals heterogeneity across treatment groups: for example, at the age of first childbirth, parenthood accounts for 84\% of the observed gender earnings gap for $D=26$ but only 59\% for $D=30$. 
Finally, I discuss aggregation across multiple treatment groups and present two subtleties: the conventional aggregate implicitly weights treatment groups by their counterfactual earnings, and differences in fertility-timing distributions confound comparisons of aggregate estimates across strata, such as across countries.

There are multiple avenues for future research. 
For example, the analysis relies on both SUTVA and Assumption~\ref{A.no.ant.1}, and exploring how violations of these assumptions affect the results is an important next step.
Moreover, the discussion focused on earnings as the outcome. Extending the analysis under NTD to other outcomes, such as employment, hours worked, or firm-level dynamics, is another important extension. 
 
\printbibliography

\appendix

\renewcommand{\thefigure}{\Alph{section}\arabic{figure}}
\renewcommand{\thetable}{\Alph{section}\arabic{table}}
\renewcommand{\theequation}{\Alph{section}\arabic{equation}}
\setcounter{figure}{0}
\setcounter{table}{0}
\setcounter{equation}{0}

\gdef\thesection{\Alph{section}}

\section{The Child-Penalty Event-Study Literature}\label{sec:app_validation_lit}

This appendix documents how the child-penalty literature validates the gender-differenced normalized event study analyzed in this paper. I find that papers support a causal reading with between-gender parallel pre-trends.

First, I compiled a broad, non-exhaustive set of empirical papers that estimate child penalties by comparing mothers and fathers in a normalized event study around the first birth, following \citet{kleven2019children}. Not included are papers that identify child penalties by other means, e.g., parents versus childless individuals, birth-cohort comparisons, within-couple (spousal) differences, instrument-based designs, and pseudo-panel constructions, since the identification results in this paper concern the between-gender normalized event study comparison specifically.

Below, for each paper, I present the quote that discusses the validation test for the between-gender normalized event study empirical strategy. Generally, the validation test is that the trajectories of pre-childbirth normalized estimates are the same across mothers and fathers. Hence, this is the testing procedure that I study when analyzing the identification assumptions in Section~\ref{sec:iden_ass}.

\begin{itemize}
\item \citet{kleven2019children}: ``We see that, once life-cycle and time trends are taken out, the earnings of men and women evolve in an almost parallel fashion until parenthood. But at the precise moment the first child arrives, the earnings paths of men and women diverge'' (p.~189).
\item \citet{kleven2019child}: ``In each country, the earnings of men and women evolve similarly before parenthood---after adjusting for life cycle and time trends---but diverge sharply after parenthood'' (p.~123).
\item \citet{kleven2021does}: ``Consider first biological families. Relative to the underlying life cycle and time trends, the earnings of men and women evolve in parallel until childbirth and then diverge sharply'' (p.~189); ``Consider then adoptive families. The main insight from panel A of Figure 1 is that adoptive families are affected by parenthood in much the same way as biological families. The earnings of adoptive parents evolve in parallel before having children and then diverge sharply and persistently after having children'' (pp.~190--191).
\item \citet{cortes2023children}: ``Importantly, men's and women's earnings evolve similarly in the years prior to parenthood, supporting the causal interpretation that the earnings losses experienced by women are indeed the result of children'' (p.~1364).
\item \citet{gould2024child}: ``Figure 1 illustrates the estimated impact of the first child on men's and women's earnings. Panel A shows that, according to the administrative panel data, the trends in earnings (relative to the base period in $t=-1$) for men and women are similar before the onset of parenthood'' (p.~12).
\item \citet{dequinto2021child}: ``Hence, the separate earnings dynamics for mothers and fathers should be interpreted with caution. The second step (the child penalty), by contrast, relies on the milder smoothness assumption of event studies. This assumption is validated with tests of parallel trends before childbirth, i.e., the wage (or other outcome) trajectories of men and women being parallel for $t < 0$. Even under parallel trends, two remarks are in order: first, this approach requires no anticipation, as women who anticipate that next-year's earnings will decline might decide (to try) to become pregnant; second, the smoothness around the date of birth may be less informative as we consider periods further away from the event'' (p.~595).
\item \citet{sieppi2019parenthood}: ``According to Kleven et al. (2018), the approach provides a plausible method to identify the causal effect of parenthood'' (p.~6); ``The presence of parallel pre-trends in outcomes provides robustness for the analysis that uses men as a control group for women'' (p.~6).
\item \citet{hotz2017parenthood}: ``As one can see in Figure 1-A, there is not a sizable difference in the trends in average wage rates across gender in the years preceding first births, but immediately after the first birth women's wages fall behind males' wages (which do not change after having their first child)'' (p.~13).
\item \citet{bazen2021measuring}: ``Figures 3a to 3d present the effects of parenthood on outcomes in the labour market, along with a 95\% confidence interval. The presence of parallel pre-trends in outcomes provides robustness for the analysis that uses men as a control group for women. In the years prior to the birth of the first child there is no significant difference between the trends of the variable for future mothers and fathers'' (pp.~10--11).
\item \citet{meurs2019gender}: ``Ideally, the estimated coefficients for this panel should be zero before year zero (no pre-trend), which would lend weight to the interpretation that observed divergences would be only linked the `birth' event (cf.\ Box 3). This condition is not met (the coefficients are significantly different from zero), neither for women nor for men. However, the trends for women and men are parallel, suggesting that the widening gap between women and men is indeed related to the entry into parenthood. Moreover, the magnitude of the pre-trend is very small compared to the changes that follow the event'' (p.~125).
\item \citet{rabate2022determines}: ``To evaluate whether the exogeneity of timing is a plausible assumption, we examine the pre-trends in the event study'' (p.~202); ``For all outcomes considered (earnings in panel A, employment in panel B, hours worked in panel C, wage rate in panel D), fathers and mothers exhibit a similar pattern before the birth. Even if the level is higher for men for all outcomes except for employment, the evolution is roughly similar for fathers and mothers. Moreover, the pre-birth trends seem to be stable, providing first evidence that there is no adjustment in labour market outcomes pre-birth that would be problematic for identification'' (p.~202).
\item \citet{berniell2021gender}: ``Although employment trajectories exhibit positive pretrends, the estimated coefficients are not statistically significant for mothers before $\tau = 0$. In any case, pretrends of men and women are parallel, and while the positive trend for men continues for a couple of years after becoming fathers, for women there is a drastic break that coincides with the year of birth of the first child. As discussed in Section 2.2, this suggests that if we were willing to think of fathers as the control group for mothers, our results underestimate the impact of motherhood on women's employment'' (p.~9, fn.~21).
\item \citet{butikofer2018role}: ``We formally test whether the prebirth trends are parallel for each panel in Figs.\ 5 and 6. Note that Fig.\ 6 is based on the sample that we use in the event-study analyses. For each graduate degree, we interact gender with each of the three prebirth periods and jointly test whether the coefficients are significantly different across the prebirth periods. None of the F-test statistics indicates that the trends are significantly different'' (p.~109).
\item \citet{zhang2024gender}: ``Fig.\ 1 presents the gender-specific effects of childbirth on individual earnings, labor force participation, working hours and wage rate over a 12-year period around the first child's birth. As defined in Eq.\ (1), these coefficients indicate outcomes at year $t$ relative to the year before the first childbirth ($\tau = -1$). This figure includes 95 percent confidence intervals around the event coefficients. Panel A shows that the earnings of men and women evolve in almost parallel trends before the first childbirth and diverge after the first childbirth, with the difference continuing for several years'' (p.~283).
\end{itemize}

\section{Survey Data: UKHLS and SOEP}\label{sec:app_survey_data}

This appendix documents the two survey data sources used to assess the generalizability of the main empirical findings outside the Israeli administrative context: Understanding Society (UKHLS) for the United Kingdom and the Socio-Economic Panel (SOEP) for Germany.

\subsection{United Kingdom Data}\label{sec:app_ukhls}

The UK application uses survey data from Understanding Society (UKHLS), a longitudinal household study administered by the Institute for Social and Economic Research \citep{ukhls2024}. The dataset combines 15 waves of UKHLS (2009--2024) with 18 waves of the Harmonised British Household Panel Survey (BHPS, 1991--2008), linked through a common cross-wave person identifier.

\subsubsection{Outcome and Treatment Variables}
The outcome variable is annualized gross labor income, constructed by multiplying the derived monthly gross labor income variable by twelve and deflating to real 2020 GBP using a national CPI index. For BHPS waves in which this derived variable is unavailable, the usual gross pay variable is used instead. Age at first childbirth is computed using two sources. The preferred source is the respondent's fertility history, recorded in the \texttt{natchild} module (UKHLS waves 1 and 6) and \texttt{childnt} module (BHPS waves 2, 11, 12), which list all biological children regardless of co-residence. For respondents not covered by the fertility history module, the household child roster (\texttt{child.dta}) is used as a fallback. The fertility history is preferred over the child roster since the child roster records only co-resident children aged 0--15 and may miss non-co-resident or older children, potentially recording a later-born child as the first and biasing $D$ upward.

\subsubsection{Sample Restrictions}
The sample is restricted to parents whose first child was born between ages 20 and 36, observed in survey waves with non-missing earnings. Unlike the administrative data, where individuals not appearing in tax records can be assigned zero earnings, survey non-response does not imply zero earnings---respondents may be employed but absent from a given wave due to attrition or item non-response. The sample therefore conditions on being observed with a valid earnings report. Following \citet{costadias2020}, two further restrictions are applied: individuals ever resident in Northern Ireland are excluded, as Northern Ireland entered the UKHLS sampling frame only in 2001, creating a sample-frame discontinuity; and individuals who ever report long-term sickness or disability (\texttt{jbstat}$=8$) are dropped. Person-years below age 18 are also excluded. Earlier treatment groups are included in the UK data as a large part of the treatment distribution is in earlier ages. In the UK, the mode of age at first childbirth for mothers is approximately 22--25, with substantial mass at ages 20--22. In Israel the distribution is more shifted to the right, with a mode of 28--29 (Figure~\ref{fig:d_dist_israel}). Accordingly, the estimation sample includes treatment groups $D\in\{20,\ldots,30\}$, starting at 20 rather than 24 (as in the main analysis) to capture a larger share of the UK treatment-group distribution. Larger values of age at first childbirth yield small cell sizes, and hence the analysis goes up to treatment group $D=30$ evaluated five years post treatment. The broader $D\in\{20,\ldots,36\}$ window is retained to supply the control groups required for estimation. The final panel comprises approximately 354 thousand person-year observations from around 40 thousand individuals.

\subsubsection{Human Capital Variables}
Several measures of human capital and family background are used to assess selection on treatment timing (Appendix Figure~\ref{fig:app_desc_cross}, top row). \textit{School leaving age} (\texttt{scend\_dv}) records the age at which the respondent left full-time education. Two measures of paternal background are taken from the cross-wave file: \textit{father's education} (\texttt{paedqf}), coded on a five-point scale where 5 denotes a university degree, from which a binary indicator for holding a degree is constructed; and \textit{father's occupation at age 14} (\texttt{pasoc90\_cc}), the father's SOC-90 occupational code, from which an indicator for managerial or professional occupations (codes 10--29) is constructed. School leaving age has high coverage, with non-missing values for 97\% of the estimation sample. The paternal background variables have lower coverage (70--79\%), as they are collected retrospectively and depend on respondent recall.

\subsection{Germany Data}\label{sec:app_soep}

The German application uses survey data from the Socio-Economic Panel (SOEP-CORE v41eu), a longitudinal household study administered by DIW Berlin \citep{soep2024}.

\subsubsection{Outcome and Treatment Variables}
The outcome variable is annualized gross labor income, constructed by multiplying current gross monthly labor income (\texttt{pglabgro}) by twelve and deflating to real 2020 EUR using a national CPI index. Non-employed person-years (employment status \texttt{pgemplst}$=5$) are assigned zero earnings, mirroring the UK and Israeli conventions; other negative values of \texttt{pglabgro} are treated as missing. Age at first childbirth is computed from the SOEP fertility biography file (\texttt{biobirth.dta}), which records the birth year of each biological child (\texttt{kidgeb01}--\texttt{kidgeb19}). The earliest non-missing child birth year is taken as the first child, and $D$ is defined as the first-birth year minus the respondent's year of birth.

\subsubsection{Sample Restrictions}
The sample is restricted to West Germany only, identified by the modal sampling-region indicator (\texttt{sampreg}$=1$) over an individual's panel records, following \citet{kleven2019child}. This restriction reflects structural pre-1990 differences in female labor force participation between East and West Germany; pooling the two regimes would mix two distinct labor market structures. The sample is further restricted to birth cohorts 1950--1984 for two reasons: younger cohorts are right-censored, while older cohorts are affected by the Second World War. Person-years with age below 18, missing earnings, or with the individual ever observed in a sheltered workshop (\texttt{pgemplst}$=7$, the SOEP analogue of long-term sick or disabled status) are excluded, mirroring the corresponding restriction applied to the UK data. As in the UK, earlier treatment groups are included given the German first-birth distribution: the estimation sample uses treatment groups $D\in\{20,\ldots,30\}$, with the broader window $D\in\{20,\ldots,36\}$ retained to provide control groups. The final panel comprises approximately 258 thousand person-year observations from around 36 thousand individuals.

\subsubsection{Human Capital Variables}
Several measures of human capital and family background are used to assess selection on treatment timing (Appendix Figure~\ref{fig:app_desc_cross}, middle and bottom rows). Cognitive ability is captured by three tests from the SOEP \texttt{cognit} module: symbol-digit substitution speed (\texttt{f99z90s}), word recall accuracy (\texttt{f25r}), and verbal fluency (\texttt{f96t90s}). Two measures of paternal background are taken from the biographical parent file (\texttt{bioparen}), which records respondents' retrospective reports on their parents during youth: an indicator for the father holding a university-track degree (\texttt{fprofedu}$\in\{5,6,11\}$, covering Fachhochschule, Universit\"at, and pre-1990 Hochschule codes); and the father's ISEI-88 occupational prestige score (\texttt{fisei88}). Maternal background variables are available but display less informative gradients in these older cohorts and are omitted.
\section{Data}\label{sec:appendix_data}

\subsection{Variable Definitions}\label{sec:appendix_data_vars}

\textit{Grandparents earnings rank.} For each parent in the data, I link the identifiers of her biological father and mother, denoted as the grandfather and grandmother.
For each parent, I sum the grandparents’ total annual earnings over the years in which the parent was aged 5–10 and divide this sum by six to obtain mean household earnings during that period.
I then rank mean household earnings within the parent’s birth cohort to construct the grandparents’ earnings-rank variable.
This procedure follows standard practice in the intergenerational income-mobility literature \citep[e.g.,][]{chetty2014land}.

\textit{Years of education and highest degree.} The CBS maintains an annually updated dataset recording individuals’ years of schooling. I use this to construct a variable for the maximum observed years of education for each individual. I also construct a binary indicator for whether the individual holds a bachelor’s degree or higher.

\textit{Meitzav test score.} The Meitzav is a national standardized exam administered by the Israeli Ministry of Education in four subjects: science, mathematics, English, and Hebrew, first implemented in 2002.
For each individual with available data, I retain the mathematics score and standardize it to have mean zero and standard deviation one within each exam year.

\textit{High-school credits.} In Israel, high-school students complete subjects at varying levels, defined by the number of credit units in each subject up to a maximum of five.
For each individual, I construct binary indicators for completing five credit units in selected subjects: mathematics, English, computers, and physics, as reported in the individual’s matriculation certificate.

\textit{Ethnicity and religion.} The Civil Registry classifies individuals as Jewish or Arab. Among Jewish individuals, religious affiliation is inferred (with some measurement error) from school type, which falls into one of three streams: Ultra-Orthodox (Haredi), state-religious, and state-secular.

\subsection{Analysis Dataset Definition}\label{sec:appendix_data_def}

The main sample restrictions are described in Section~\ref{sec:data}.
Two additional exclusions, omitted from the main text due to their negligible impact on sample size, are applied throughout the analysis.
First, I drop cases in which the recorded birth year of the parent is later than that of their first child (951 observations).
Second, I drop individuals recorded as giving birth at age ten or younger (443 observations).

The definitions below describe the construction of the auxiliary datasets used for Figure~\ref{fig:desc_all} in Section~\ref{sec:did_violation}.
Because these figures do not rely on observed earnings, the lower threshold of the treatment-group range is decreased to 20, as compared to 24 in the main analysis.

\textit{Grandparents earnings rank.}
Grandparents’ earnings come from administrative income data for 1990–2020.
Because earnings data begin in 1990, the youngest parent birth cohort for which the grandparents’ earnings rank can be constructed is 1985.
Since the parent sample extends up to the 1990 birth cohort, the dataset used for this variable includes cohorts 1985–1990.
Given that income data end in 2020, the latest observable first-birth age (treatment group) for which grandparents’ earnings rank is available is 35.

\textit{Grandparents Education.} For each grandparent, I take the reported years of education and classify it into three categories: missing, at most 12 years (high school or less), or above 12 years (post-secondary education).
This variable is available for all treatment groups with first-birth ages between 20 and 40.
The share of missing observations ranges between 1–2\% for grandmothers and 3–5\% for grandfathers across treatment groups.

\textit{Meitzav.} Data on Meitzav test scores are available for the years 2002–2019.
The exams are administered in 5th and 8th grades, corresponding to ages 10–11 and 13–14, respectively. The 5th-grade test is observable for birth cohorts 1991 and onward. To analyze 5th-grade scores, I therefore include two additional cohorts not used in the main analysis sample: 1991 and 1992. Given that birth data is available up to 2020, this measure is available for treatment groups with first-birth ages 20–29.
The 8th-grade test is observable for birth cohorts 1988 and onward. To analyze 8th-grade scores, I limit the dataset to birth cohorts 1988–1990, and hence the measure is available for treatment groups aged 20–32.

\textit{High-school credits.} Data on the number of credit units in high-school subjects are available for most of the school cohort beginning with the 1980 birth cohort.
Accordingly, this measure is constructed for all treatment groups with first-birth ages between 20 and 40.
Note that sample size declines sharply at higher first-birth ages: the number of observations decreases from 25{,}710 for treatment age 30, to 4{,}451 for treatment age 35, and to 248 for treatment age 40.

\section{Selection on Fertility Timing: UK and Germany}\label{sec:app_hc_evidence}

Figure~\ref{fig:app_desc_cross} documents pre-childbirth selection on observables in the UK and German survey data described in Appendix~\ref{sec:app_survey_data}, paralleling the Israeli evidence in Section~\ref{sec:did_violation}. In the UK, parents who delay first childbirth leave full-time education at later ages and come from families with higher paternal education and a higher probability of the father being in a managerial or professional occupation. In Germany, parents who delay childbirth have fathers with higher educational attainment and occupational prestige, and they themselves score higher on three adult cognitive tests (symbol-digit speed, word recall, verbal fluency). Gradients are positive and similar across genders in both countries. The cognitive scores are measured post-childbirth for parents, and may therefore be affected by parenthood itself, but still provide suggestive evidence that inherent cognitive ability differs by age at first childbirth.
These patterns are in line with the selection on fertility mechanism discussed in Section~\ref{sec:did_violation}, and hence suggest such patterns are quite general.

\begin{figure}[p]
    \centering
    \includegraphics[width=\textwidth]{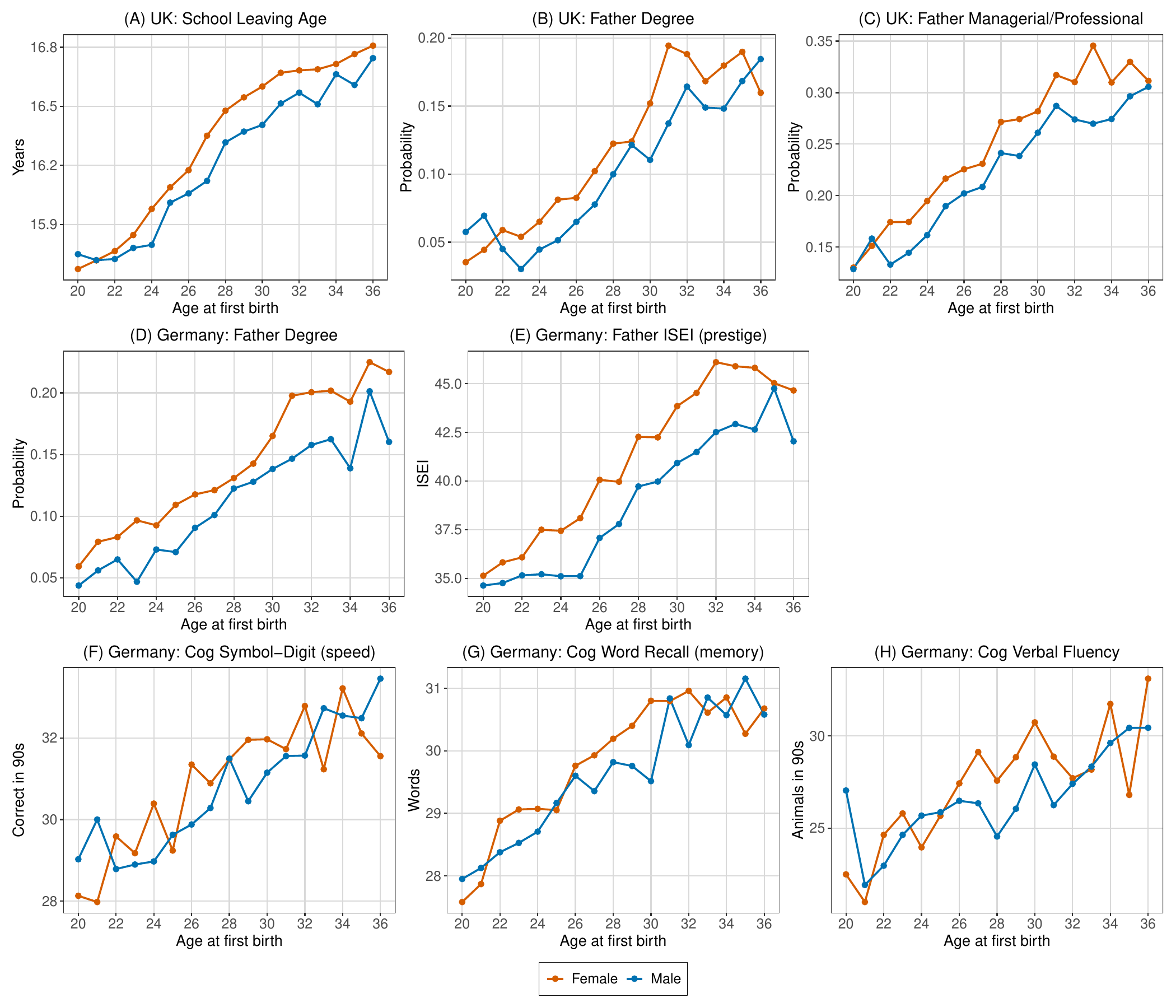}
    \caption{Pre-Childbirth Human Capital and Age at First Childbirth: UK and Germany. \textit{Notes:} The figure shows means by age at first birth (x-axis) and gender (colors). Top row, UK (Understanding Society + Harmonised BHPS): school leaving age (Panel A), probability that the father holds a university degree (Panel B), probability that the father is in a managerial or professional occupation (Panel C). Middle row, Germany (SOEP-CORE v41eu), paternal background from the biographical parent file (\texttt{bioparen}, retrospective reports about the father during the respondent's youth): probability that the father holds a university-track degree (Panel D), father's ISEI-88 occupational prestige (Panel E). Bottom row, Germany, cognitive tests from the SOEP \texttt{cognit} module: symbol-digit substitution speed (Panel F), word recall (Panel G), verbal fluency (Panel H). All measures are raw means. Sample definitions and variable details are in Appendix~\ref{sec:app_survey_data}.}
    \label{fig:app_desc_cross}
\end{figure}

\section{Proportional DID and TD}\label{sec:td_violation}

This appendix discusses proportional parallel-trend-type assumptions and their triple difference counterparts, and compares these to normalized triple differences (NTD). I first consider parallel trends (PT) within gender under a log functional form, in multiplicative form, and in logs of expectations, and argue that the human capital selection mechanism presented in Section~\ref{sec:did_violation} suggests all of these assumptions are violated. Then, for each type of PT assumption, I consider its triple difference (TD), which allows violations but assumes they are the same across genders, and discuss what each of these proportional TD assumptions identifies. Finally, I compare estimates under each approach empirically using the Israeli data. I find that they perform similarly in pre-treatment validation tests. Post-treatment estimates largely coincide at childbirth and diverge at longer event-time horizons.

\subsection{Different PT Assumptions}

Recall $\gamma_{\mathrm{PT}}(g,d,d^\prime,a)$ is the difference in counterfactual trends between treatment group $d$ and control group $d^\prime$ for gender $g$, from age $d-1$ to age $a$. For this appendix, I denote this quantity by $\gamma_{\text{PT-LEVELS}}(g,d,d^\prime,a)$, and for brevity denote $\gamma_{\text{PT-LEVELS}}(g)\equiv\gamma_{\text{PT-LEVELS}}(g,d,d^\prime,a)$. For completeness, I restate the parallel trends assumption in levels. 
\begin{assumption}[PT-LEVELS]{}\label{A.pt.levels}
For gender $g$ and a triplet $(d,d^\prime,a)$ with $d\leq a<d^\prime$, $\gamma_{\text{PT-LEVELS}}(g)=0$.
\end{assumption}
\noindent
Next, I consider different ways to write the difference in trends in proportional terms. The direct approach is to transform the outcome using logs. Let
\begin{align*}
\gamma_{\text{PT-LOGS}}(g,d,d',a) & =\mathbb{E}[\log\left(Y_{a}\left(\infty\right)\right)-\log\left(Y_{d-1}\left(\infty\right)\right)\mid G=g,D=d]\\
 & -\mathbb{E}[\log\left(Y_{a}\left(\infty\right)\right)-\log\left(Y_{d-1}\left(\infty\right)\right)\mid G=g,D=d']\\
 & =\mathbb{E}\left[\log\left(\frac{Y_{a}(\infty)}{Y_{d-1}(\infty)}\right)\mid G=g,D=d\right]\\
 & \quad-\mathbb{E}\left[\log\left(\frac{Y_{a}(\infty)}{Y_{d-1}(\infty)}\right)\mid G=g,D=d'\right].
\end{align*}
As before, for brevity denote $\gamma_{\text{PT-LOGS}}(g)\equiv\gamma_{\text{PT-LOGS}}(g,d,d',a)$. The parallel trends assumption under a log functional form can be written as follows. 
\begin{assumption}[PT-LOGS]{}\label{A.pt.logs}
For gender $g$ and a triplet $(d,d^\prime,a)$ with $d\leq a<d^\prime$, $\gamma_{\text{PT-LOGS}}(g)=0$.
\end{assumption}
In the considered context, using Assumption \ref{A.pt.logs} has the disadvantage that earnings data include many zeros, which around childbirth may be affected by treatment. Alternatives such as using log plus one introduce other caveats \citep{chen2024logs}. Due to these issues, I discuss two other approaches. 
\citet{ciani2019dif} discuss multiplicative instead of additive parallel trends, based on the difference in proportional change. Let
\[
\gamma_{\text{PT-PR}}(g,d,d',a)=\frac{\mathbb{E}[Y_{a}\left(\infty\right)\mid G=g,D=d]}{\mathbb{E}[Y_{d-1}\left(\infty\right)\mid G=g,D=d]}-\frac{\mathbb{E}[Y_{a}\left(\infty\right)\mid G=g,D=d^{\prime}]}{\mathbb{E}[Y_{d-1}\left(\infty\right)\mid G=g,D=d^{\prime}]}.
\]
For brevity denote $\gamma_{\text{PT-PR}}(g)\equiv \gamma_{\text{PT-PR}}(g,d,d',a)$. Similarly to $\gamma_{\text{PT-LOGS}}(g)$, $\gamma_{\text{PT-PR}}(g)$ compares proportional changes in mean earnings over time between treatment and control groups. However, it does not require the log transformation of the outcome, and hence easily incorporates zeros. A similar approach builds on \citet{wooldridge2023simple}, which takes logs of the expectations of the outcomes. Let
\begin{align*}
\gamma_{\text{PT-E-LOGS}}(g,d,d',a) & =\log\mathbb{E}[Y_{a}\left(\infty\right)\mid G=g,D=d]-\log\mathbb{E}[Y_{d-1}\left(\infty\right)\mid G=g,D=d]\\
 & -\left(\log\mathbb{E}[Y_{a}\left(\infty\right)\mid G=g,D=d^{\prime}]-\log\mathbb{E}[Y_{d-1}\left(\infty\right)\mid G=g,D=d^{\prime}]\right)\\
 & =\log\left(\frac{\mathbb{E}[Y_{a}\left(\infty\right)\mid G=g,D=d]/\mathbb{E}[Y_{d-1}\left(\infty\right)\mid G=g,D=d]}{\mathbb{E}[Y_{a}\left(\infty\right)\mid G=g,D=d^{\prime}]/\mathbb{E}[Y_{d-1}\left(\infty\right)\mid G=g,D=d^{\prime}]}\right)
\end{align*}
As above, for brevity denote $\gamma_{\text{PT-E-LOGS}}(g)\equiv \gamma_{\text{PT-E-LOGS}}(g,d,d',a)$. Similar to the multiplicative approach, taking logs of expectations accommodates outcomes with zeros. Note that $\gamma_{\text{PT-PR}}(g)=0$, i.e., treatment and control groups have the same proportional trend over time, if and only if $\gamma_{\text{PT-E-LOGS}}(g)=\log(1)=0$. Hence, these can be combined into a single assumption. 
\begin{assumption}[PT-PR]{}\label{A.pt.pr}
For gender $g$ and a triplet $(d,d^\prime,a)$ with $d\leq a<d^\prime$, $\gamma_{\text{PT-PR}}(g)=\gamma_{\text{PT-E-LOGS}}(g)=0$.    
\end{assumption}

\subsection{Within Gender Violations}

Section~\ref{sec:did_violation} argues that the parallel trends assumption in levels likely fails due to selection on treatment timing. Does the same selection mechanism also invalidate parallel trends in proportional terms?
For completeness, I briefly restate the argument from the main text. Human capital theory suggests that later-treated individuals have higher labor-market ability, invest more in human capital and so enter the labor market later, and then experience steeper absolute earnings growth. This suggests that, at early ages, the later-treated control group has lower absolute earnings than the earlier-treated group, which is more established in the labor market at those ages. 

Faster absolute growth from a lower base implies faster proportional growth. Hence, if both patterns hold, the parallel trends assumption is violated in proportional terms as well. 
Consider a numerical example. All quantities are mean counterfactual earnings absent childbirth. Compare two treatment groups at baseline age 24 (one year before treatment for $d=25$) and target age 29; at age 29, the closest group not yet treated is $d=30$:
\begin{itemize}
\item Treatment group ($d=25$): $Y_{24} = 100$, growing to $Y_{29} = 120$.
\item Control group ($d=30$): $Y_{24} = 50$, growing to $Y_{29} = 130$.
\end{itemize}
In levels, counterfactual earnings grow by $20$ in the treatment group and by $80$ in the control group, so $\gamma_{\text{PT-LEVELS}} = 20 - 80 = -60$: a violation of the parallel trends assumption in levels, such that the control group has a steeper trend than the treatment group. In proportional terms, the growth factors and their logs are
\begin{itemize}
\item Treatment group ($d=25$): $Y_{29}/Y_{24} =120/100 = 1.20$, and $\log(120/100) \approx 0.182$.
\item Control group ($d=30$): $Y_{29}/Y_{24} = 130/50 = 2.60$, and $\log(130/50) \approx 0.956$.
\end{itemize}
so $\gamma_{\text{PT-PR}} = 1.20 - 2.60 = -1.40$ and $\gamma_{\text{PT-E-LOGS}} \approx 0.182 - 0.956 \approx -0.77$. The control group grows faster in proportional terms as well, because it starts from a lower base while experiencing larger absolute growth.

\subsection{Proportional TD Assumptions and Identification}

Given the above discussion that the parallel trends assumptions, both in levels and proportional, are violated, I turn to discussing TD-type assumptions, which allow violations but assume equal violations across genders. For completeness, I first restate the TD assumption in levels and the NTD assumption from the main text. 

\begin{assumption}[TD-LEVELS]{}\label{A.td.levels}
For a triplet $(d,d^\prime,a)$ with $d\leq a<d^\prime$, $\gamma_{\text{PT-LEVELS}}(f)=\gamma_{\text{PT-LEVELS}}(m)$.
\end{assumption}

\begin{assumption}[NTD]{}\label{A.ntd.oa}
For a triplet $(d,d^\prime,a)$ with $d\leq a<d^\prime$, $\APO(f,\infty),\APO(m,\infty)\neq0$ and
$$\frac{\gamma_{\text{PT-LEVELS}}(f)}{\APO(f,\infty)}=\frac{\gamma_{\text{PT-LEVELS}}(m)}{\APO(m,\infty)}.$$
\end{assumption}
\noindent
Next, I consider TD-type assumptions on the proportional differences in trends $\gamma_{\text{PT-PR}}(g)$, and on the log of expectations $\gamma_{\text{PT-E-LOGS}}(g)$. 

\begin{assumption}[TD-PR]{}\label{A.td.pr}
For a triplet $(d,d^\prime,a)$ with $d\leq a<d^\prime$, $\gamma_{\text{PT-PR}}(f)=\gamma_{\text{PT-PR}}(m)$.
\end{assumption}

\begin{assumption}[TD-E-LOGS]{}\label{A.td.elogs}
For a triplet $(d,d^\prime,a)$ with $d\leq a<d^\prime$, $\gamma_{\text{PT-E-LOGS}}(f)=\gamma_{\text{PT-E-LOGS}}(m)$.
\end{assumption}

A first question is how these assumptions are related. Let $\mu_{g,d,a}=\mathbb{E}[Y_{a}\mid G=g,D=d]$ denote mean observed earnings, and recall the DID-imputed counterfactual $\delta_{\APO}(g)=\mu_{g,d,d-1}+(\mu_{g,d^\prime,a}-\mu_{g,d^\prime,d-1})$, which imputes counterfactual earnings additively. Analogously, define the multiplicative imputation
$$\delta^{\times}_{\APO}(g,d,d^\prime,a)=\mu_{g,d,d-1}\cdot\frac{\mu_{g,d^\prime,a}}{\mu_{g,d^\prime,d-1}},$$
and for brevity $\delta^{\times}_{\APO}(g)\equiv\delta^{\times}_{\APO}(g,d,d^\prime,a)$.
Under no anticipation, i.e., Assumption~\ref{A.no.ant.1}, and since control group $d^\prime$ is not yet treated at ages $d-1$ and $a$, the trend differences can be written as imputation errors:
\begin{align*}
\gamma_{\text{PT-LEVELS}}(g) &= \APO(g,\infty)-\delta_{\APO}(g),\\
\gamma_{\text{PT-E-LOGS}}(g) &= \log\big(\APO(g,\infty)/\delta^{\times}_{\APO}(g)\big),\\
\gamma_{\text{PT-PR}}(g) &= \big(\APO(g,\infty)-\delta^{\times}_{\APO}(g)\big)/\mu_{g,d,d-1}.
\end{align*}
Hence the TD-type assumptions restrict imputation errors in different ways. Assumption~\ref{A.td.levels} requires the additive imputation error to be equal across genders in levels, while Assumption~\ref{A.ntd.oa} requires it to be equal as a share of counterfactual earnings,
$$\frac{\APO(f,\infty)-\delta_{\APO}(f)}{\APO(f,\infty)}=\frac{\APO(m,\infty)-\delta_{\APO}(m)}{\APO(m,\infty)}.$$
Assumption~\ref{A.td.elogs} is the multiplicative analog of NTD: it requires the multiplicative imputation error to be equal across genders as a share of counterfactual earnings,
$$\frac{\APO(f,\infty)-\delta^{\times}_{\APO}(f)}{\APO(f,\infty)}=\frac{\APO(m,\infty)-\delta^{\times}_{\APO}(m)}{\APO(m,\infty)}.$$
Assumption~\ref{A.td.pr} instead requires the multiplicative imputation error to be equal across genders when scaled by baseline earnings $\mu_{g,d,d-1}$, rather than by counterfactual earnings.

The following proposition states what the two new proportional TD-type assumptions identify. Recall that $\rho(d,d^\prime,a)=\frac{\APO(f,d,d^\prime,a)}{\APO(m,d,d^\prime,a)}$ denotes the gender ratio of mean potential earnings.

\begin{appprop}\label{prop:td_prop_identification}
Consider a triplet $(d,d^\prime,a)$ with $d\leq a<d^\prime$, assume no anticipation holds and that $\APO(g,\infty)$, $\mu_{g,d,d-1}$, $\mu_{g,d^\prime,d-1}$ and $\mu_{g,d^\prime,a}$ are positive for $g\in\{f,m\}$.
\begin{enumerate}[(i)]
\item Under Assumption~\ref{A.td.pr},
\begin{align*}
\frac{\ATE(f,d,a)}{\APO(f,d,\infty,d-1)}-\frac{\ATE(m,d,a)}{\APO(m,d,\infty,d-1)}&=\frac{\mu_{f,d,a}-\delta^{\times}_{\APO}(f,d,d^\prime,a)}{\mu_{f,d,d-1}}\\
&\quad -\frac{\mu_{m,d,a}-\delta^{\times}_{\APO}(m,d,d^\prime,a)}{\mu_{m,d,d-1}}.
\end{align*}
\item Under Assumption~\ref{A.td.elogs},
$$\rho(d,\infty,a)=\frac{\delta^{\times}_{\APO}(f,d,d^\prime,a)}{\delta^{\times}_{\APO}(m,d,d^\prime,a)}.$$
\end{enumerate}
\end{appprop}

\begin{proof}
Fix $g\in\{f,m\}$. By consistency, $\APO(g,d,d,a)=\mu_{g,d,a}$, so $\ATE(g,d,a)=\mu_{g,d,a}-\APO(g,d,\infty,a)$. By no anticipation applied to treatment group $d$, $\APO(g,d,\infty,d-1)=\mu_{g,d,d-1}$, and applied to control group $d^\prime$, which is not yet treated at ages $d-1$ and $a$, $\APO(g,d^\prime,\infty,t)=\mu_{g,d^\prime,t}$ for $t\in\{d-1,a\}$. For the proportional difference in trends this implies
\begin{align*}
\gamma_{\text{PT-PR}}(g) &= \frac{\mathbb{E}[Y_{a}(\infty)\mid G=g,D=d]}{\mathbb{E}[Y_{d-1}(\infty)\mid G=g,D=d]}-\frac{\mathbb{E}[Y_{a}(\infty)\mid G=g,D=d^\prime]}{\mathbb{E}[Y_{d-1}(\infty)\mid G=g,D=d^\prime]}\\
&= \frac{\APO(g,d,\infty,a)}{\mu_{g,d,d-1}}-\frac{\mu_{g,d^\prime,a}}{\mu_{g,d^\prime,d-1}}\\
&= \frac{\APO(g,d,\infty,a)-\mu_{g,d,d-1}\cdot\mu_{g,d^\prime,a}/\mu_{g,d^\prime,d-1}}{\mu_{g,d,d-1}}\\
&= \frac{\APO(g,d,\infty,a)-\delta^{\times}_{\APO}(g,d,d^\prime,a)}{\mu_{g,d,d-1}}.
\end{align*}
Similarly, for the log of expectations,
\begin{align*}
\gamma_{\text{PT-E-LOGS}}(g) &= \log\mathbb{E}[Y_{a}(\infty)\mid G=g,D=d]-\log\mathbb{E}[Y_{d-1}(\infty)\mid G=g,D=d]\\
&\quad -\big(\log\mathbb{E}[Y_{a}(\infty)\mid G=g,D=d^\prime]-\log\mathbb{E}[Y_{d-1}(\infty)\mid G=g,D=d^\prime]\big)\\
&= \log\APO(g,d,\infty,a)-\log\mu_{g,d,d-1}-\big(\log\mu_{g,d^\prime,a}-\log\mu_{g,d^\prime,d-1}\big)\\
&= \log\APO(g,d,\infty,a)-\log\big(\mu_{g,d,d-1}\cdot\mu_{g,d^\prime,a}/\mu_{g,d^\prime,d-1}\big)\\
&= \log\frac{\APO(g,d,\infty,a)}{\delta^{\times}_{\APO}(g,d,d^\prime,a)}.
\end{align*}

For (i), adding and subtracting $\delta^{\times}_{\APO}(g,d,d^\prime,a)$ in the numerator,
\begin{align*}
\frac{\ATE(g,d,a)}{\APO(g,d,\infty,d-1)} &= \frac{\mu_{g,d,a}-\delta^{\times}_{\APO}(g,d,d^\prime,a)}{\mu_{g,d,d-1}}-\frac{\APO(g,d,\infty,a)-\delta^{\times}_{\APO}(g,d,d^\prime,a)}{\mu_{g,d,d-1}}\\
&= \frac{\mu_{g,d,a}-\delta^{\times}_{\APO}(g,d,d^\prime,a)}{\mu_{g,d,d-1}}-\gamma_{\text{PT-PR}}(g).
\end{align*}
Taking the difference between $g=f$ and $g=m$, the $\gamma_{\text{PT-PR}}(g)$ terms cancel under Assumption~\ref{A.td.pr}, which yields (i).

For (ii), substituting the above expression for $\gamma_{\text{PT-E-LOGS}}(g)$ into Assumption~\ref{A.td.elogs} gives
$$\log\frac{\APO(f,d,\infty,a)}{\delta^{\times}_{\APO}(f,d,d^\prime,a)}=\log\frac{\APO(m,d,\infty,a)}{\delta^{\times}_{\APO}(m,d,d^\prime,a)}.$$
Since the logarithm is strictly increasing, the arguments of the logarithms are equal, and rearranging gives (ii).
\end{proof}
\noindent
The result shows that under Assumption~\ref{A.td.elogs}, the effect of parenthood on the gender earnings ratio is identified, as under NTD, since $\rho(d,d,a)=\mu_{f,d,a}/\mu_{m,d,a}$ by consistency.

\subsection{Assessing the Different TD Approaches}

Intuitively, Assumptions \ref{A.ntd.oa}, \ref{A.td.pr} and \ref{A.td.elogs} make similar statements about the underlying data generating process, namely that the proportional difference in trends between treatment and control groups is equal across men and women. However, the specific restrictions differ, as discussed above, and hence may yield different estimates. Below, I first conduct validation tests separately for each assumption, and then show the differences in post-childbirth estimates.

\paragraph{Validation Tests.} For each assumption, I apply its estimator, the sample analog of the descriptive estimand from the identification formula, at pre-treatment target ages. This builds on the discussion of validation testing in Sections~\ref{sec:did_violation} and~\ref{sec:alt_estimand}. Figure~\ref{fig:td_validation} reports the results by treatment group and control group, for treatment groups $d=26,\ldots,30$, as those were the treatment groups I focused on in Section~\ref{sec:alt_estimand}. As shown in the figure, the estimates are small, with no systematic differences by the gap in treatment age between treatment and control groups. The mean absolute estimate is $0.017$ under Assumption \ref{A.ntd.oa}, $0.018$ under Assumption~\ref{A.td.pr} and $0.022$ under Assumption~\ref{A.td.elogs}. Hence, the validation tests do not indicate which assumption is more likely to hold in the data. 

\begin{figure}[t!]
    \centering
    \includegraphics[width=\textwidth]{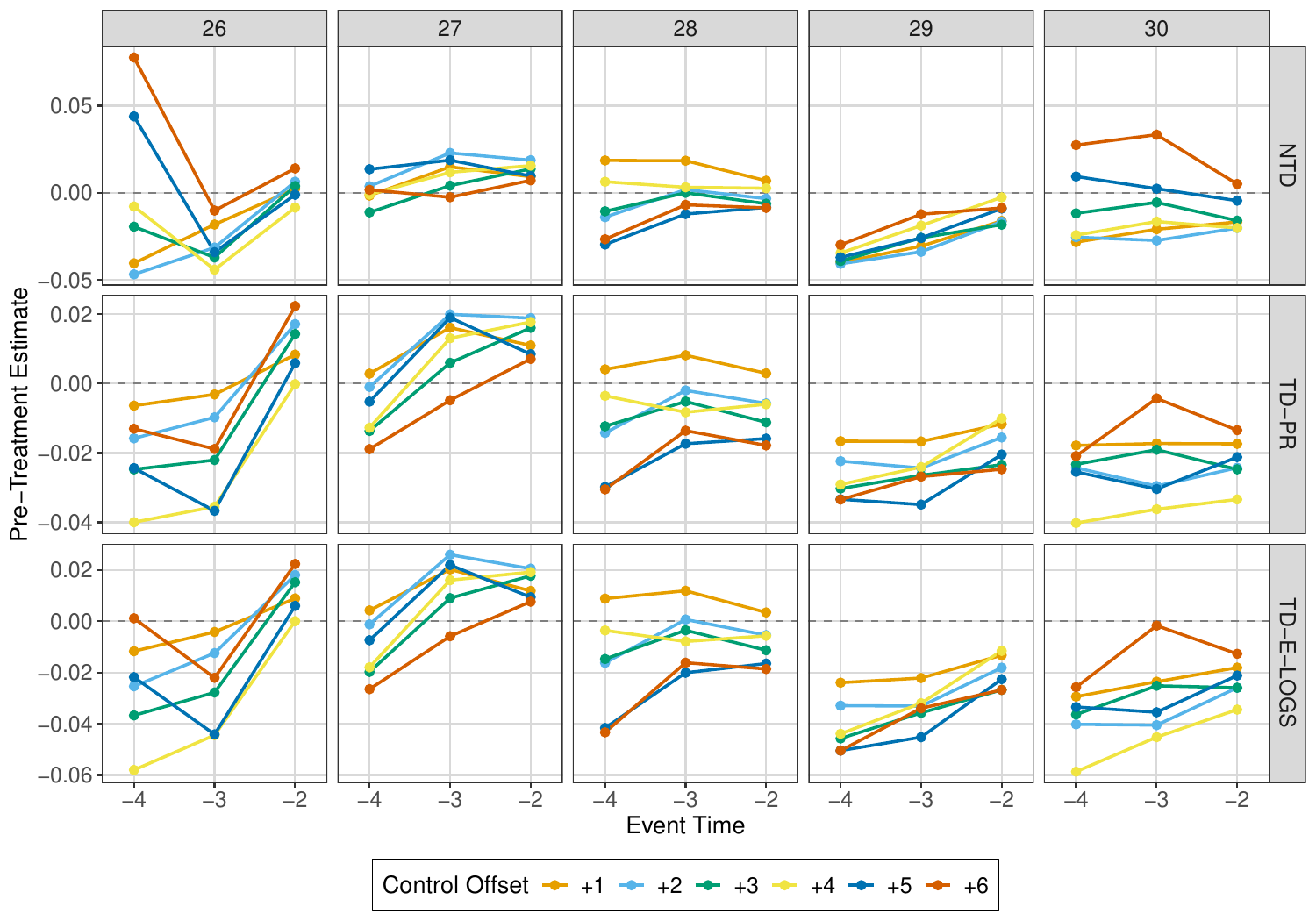}
    \caption{Pre-Treatment Validation Tests of the Proportional TD Approaches. \textit{Notes:} Columns are treatment groups $d\in\{26,\ldots,30\}$ and rows are assumptions. Each line is one control group $d^{\prime}=d+1,\ldots,d+6$ (colors, by control offset $d^{\prime}-d$), with estimates at pre-treatment event times $-4,-3,-2$ (x-axis) and baseline age $d-1$. For each assumption, the estimator is the sample analog of the descriptive estimand in the identification formula.}
    \label{fig:td_validation}
\end{figure}

\paragraph{Post-Treatment Estimates.} NTD and Assumption~\ref{A.td.elogs} identify the same estimand, the effect of parenthood on the gender earnings ratio, using the additive and the multiplicative DID-imputations, respectively. Figure~\ref{fig:td_post_estimates} shows post-childbirth estimates under each approach.\footnote{Assumption~\ref{A.td.pr} identifies a different estimand and is therefore not shown.} At childbirth, the estimates coincide: they differ by at most $0.006$ across treatment groups $d=26,\ldots,30$. However, the series diverge with event time, as the control group moves further from the treatment group. Hence the effect at childbirth does not depend on whether counterfactual earnings are imputed additively or multiplicatively, while its magnitude at longer horizons does.

\begin{figure}[t!]
    \centering
    \includegraphics[width=\textwidth]{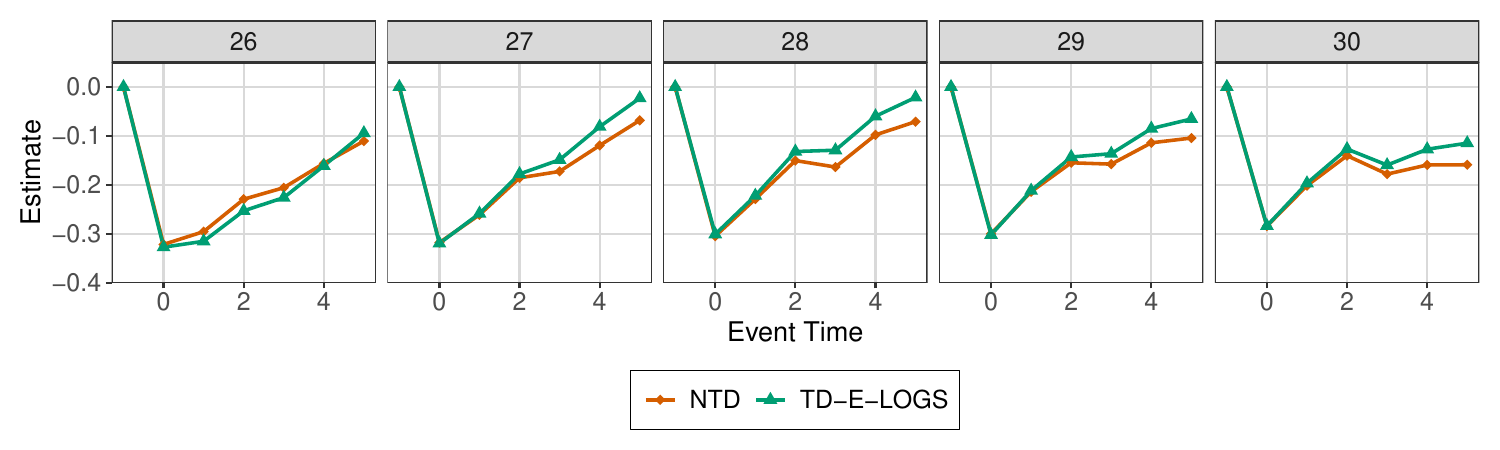}
    \caption{Effect of Parenthood on the Gender Earnings Ratio under NTD and TD-E-LOGS. \textit{Notes:} The figure presents estimates of the effect of parenthood on the gender earnings ratio under Assumption \ref{A.ntd.oa} (orange diamonds), as in Figure~\ref{fig:alt_ntd_b}, and under Assumption~\ref{A.td.elogs} (green triangles). Assumption \ref{A.ntd.oa} imputes counterfactual earnings additively and Assumption~\ref{A.td.elogs} multiplicatively. For treatment group $d$ and event time $e$, the target age is $a=d+e$ and the control group is the closest not-yet-treated group, $d^{\prime}=a+1$.}
    \label{fig:td_post_estimates}
\end{figure}

\section{Identification}\label{sec:app_ident_proofs}

\begin{proof}[Proof of Lemma \ref{lemma:desc_tie_cause}]
Re-arranging the definition of $\gamma_\mathrm{PT}$ implies
\begin{align*}
\mathbb{E}\left[Y_{a}\left(\infty\right)\mid G=g,D=d\right] & =\gamma_{\mathrm{PT}}\left(g,d,d^{\prime},a\right)+\mathbb{E}\left[Y_{d-1}\left(\infty\right)\mid G=g,D=d\right]\\
 & +\mathbb{E}\left[Y_{i,a}\left(\infty\right)-Y_{d-1}\left(\infty\right)\mid G=g,D=d^{\prime}\right].
\end{align*}
Assumption \ref{A.no.ant.1} and consistency imply
\begin{align*}
\mathbb{E}\left[Y_{a}\left(\infty\right)\mid G=g,D=d\right] & =\gamma_{\mathrm{PT}}\left(g,d,d^{\prime},a\right)+\mathbb{E}\left[Y_{d-1}\mid G=g,D=d\right]\\
 & +\mathbb{E}\left[Y_{a}-Y_{d-1}\mid G=g,D=d^{\prime}\right].
\end{align*}
Substituting the definition of $\delta_{\mathrm{APO}}$ and re-arranging we obtain \eqref{eq:app_proof_cp_ass_1}. Then substituting \eqref{eq:app_proof_cp_ass_1} into the definition of the $\delta_{\mathrm{ATE}}$ we obtain \eqref{eq:app_proof_cp_ass_2}. Finally, substitute both \eqref{eq:app_proof_cp_ass_1} and \eqref{eq:app_proof_cp_ass_2} into the definition of $\delta_{\theta}$ and re-arrange to obtain \eqref{eq:lemma_d_theta_bias}.
\end{proof}

\begin{proof}[Proof of Proposition \ref{prop:cp_ident_ass}]
Considering $\delta_{\theta}\left(f,d,d^{\prime},a\right)=\delta_{\theta}\left(m,d,d^{\prime},a\right)$
for $a<d,d^{\prime}$ we obtain
\begin{align*}
 & \frac{\delta_{\mathrm{ATE}}\left(f,d,d^{\prime},a\right)}{\delta_{\mathrm{APO}}\left(f,d,d^{\prime},a\right)}=\frac{\delta_{\mathrm{ATE}}\left(m,d,d^{\prime},a\right)}{\delta_{\mathrm{APO}}\left(m,d,d^{\prime},a\right)}\\
\overset{\text{ }}{\leftrightarrow} & \frac{\ATE\left(f,d,a\right)+\gamma_{\mathrm{PT}}\left(f,d,d^{\prime},a\right)}{\APO\left(f,d,\infty,a\right)-\gamma_{\mathrm{PT}}\left(f,d,d^{\prime},a\right)}=\frac{\ATE\left(m,d,a\right)+\gamma_{\mathrm{PT}}\left(m,d,d^{\prime},a\right)}{\APO\left(m,d,\infty,a\right)-\gamma_{\mathrm{PT}}\left(m,d,d^{\prime},a\right)}\\
\leftrightarrow & \frac{\gamma_{\mathrm{PT}}\left(f,d,d^{\prime},a\right)}{\APO\left(f,d,\infty,a\right)-\gamma_{\mathrm{PT}}\left(f,d,d^{\prime},a\right)}=\frac{\gamma_{\mathrm{PT}}\left(m,d,d^{\prime},a\right)}{\APO\left(m,d,\infty,a\right)-\gamma_{\mathrm{PT}}\left(m,d,d^{\prime},a\right)}\\
\leftrightarrow & \frac{\gamma_{\mathrm{PT}}\left(f,d,d^{\prime},a\right)}{\APO\left(f,d,\infty,a\right)}=\frac{\gamma_{\mathrm{PT}}\left(m,d,d^{\prime},a\right)}{\APO\left(m,d,\infty,a\right)},
\end{align*}
where the first iff is due to substituting \eqref{eq:app_proof_cp_ass_1} and \eqref{eq:app_proof_cp_ass_2}, the second iff follows from Assumption \ref{A.no.ant.1}, and the third follows by algebra.
\end{proof}

\begin{proof}[Proof of Theorem \ref{thm:cp_iden_res}]
Assumption \ref{A.ntd_pt} implies 
\begin{align}\label{eq:theorem_d_theta_bias_1}
\frac{\APO\left(f,d,\infty,a\right)}{\APO\left(f,d,\infty,a\right)-\gamma_{\mathrm{PT}}\left(f,d,d^{\prime},a\right)} & =\frac{\APO\left(f,d,\infty,a\right)}{\APO\left(f,d,\infty,a\right)-\frac{\APO\left(f,d,\infty,a\right)}{\APO\left(m,d,\infty,a\right)}\gamma_{\mathrm{PT}}\left(m,d,d^{\prime},a\right)}\nonumber\\
 & =\frac{1}{1-\frac{\gamma_{\mathrm{PT}}\left(m,d,d^{\prime},a\right)}{\APO\left(m,d,\infty,a\right)}}\nonumber\\
 & =\frac{\APO\left(m,d,\infty,a\right)}{\APO\left(m,d,\infty,a\right)-\gamma_{\mathrm{PT}}\left(m,d,d^{\prime},a\right)}.
\end{align}
A similar argument shows
\begin{align}\label{eq:theorem_d_theta_bias_2}
\frac{\gamma_{\mathrm{PT}}\left(f,d,d^{\prime},a\right)}{\APO\left(f,d,\infty,a\right)-\gamma_{\mathrm{PT}}\left(f,d,d^{\prime},a\right)} =\frac{\gamma_{\mathrm{PT}}\left(m,d,d^{\prime},a\right)}{\APO\left(m,d,\infty,a\right)-\gamma_{\mathrm{PT}}\left(m,d,d^{\prime},a\right)}.
\end{align}
This shows that $\mathrm{Bias}(d,d^\prime,a)$, as stated in the theorem, is indeed gender invariant. Subtracting \eqref{eq:lemma_d_theta_bias} using $g=f$ from \eqref{eq:lemma_d_theta_bias} using $g=m$, and applying \eqref{eq:theorem_d_theta_bias_1} and \eqref{eq:theorem_d_theta_bias_2} we obtain the result in the theorem, which concludes the proof.
\end{proof}

The next lemma characterizes the bias for another target causal estimand that is sometimes used in the literature to quantify child penalties. Let 
\[
\begin{aligned}
P(d,a)&=\frac{\ATE(f,d,a)-\ATE(m,d,a)}{\APO(f,d,\infty,a)},\\
\delta_P(d,d^\prime,a)&=\frac{\delta_{\mathrm{ATE}}(f,d,d',a)-\delta_{\mathrm{ATE}}(m,d,d',a)}{\delta_{\mathrm{APO}}(f,d,d',a)}.
\end{aligned}
\]

\begin{lemma}\label{lem:ntd_cross_bias}
Assume the same setup as in Theorem~\ref{thm:cp_iden_res}. 
\begin{align*}
\delta_P(d,d^\prime,a) & =P(d,a)\times\frac{\APO(f,d,\infty,a)}{\APO(f,d,\infty,a)-\gamma_{\mathrm{PT}}(f,d,d^{\prime},a)}\\
 & +\frac{\gamma_{\mathrm{PT}}(f,d,d^{\prime},a)-\gamma_{\mathrm{PT}}(m,d,d^{\prime},a)}{\APO(f,d,\infty,a)-\gamma_{\mathrm{PT}}(f,d,d^{\prime},a)}.
\end{align*}
\end{lemma}

\begin{proof}
From Lemma \ref{lemma:desc_tie_cause}
\begin{align}\label{eq:lemma2_1}
\delta_P(d,d^\prime,a) & =\frac{\ATE(f,d,a)+\gamma_{\mathrm{PT}}(f,d,d^{\prime},a)}{\APO(f,d,\infty,a)-\gamma_{\mathrm{PT}}(f,d,d^{\prime},a)} -\frac{\ATE(m,d,a)+\gamma_{\mathrm{PT}}(m,d,d^{\prime},a)}{\APO(f,d,\infty,a)-\gamma_{\mathrm{PT}}(f,d,d^{\prime},a)}.
\end{align}
Adding and subtracting $P(d,a)$ in \eqref{eq:lemma2_1},
after some algebra, yields
\begin{align*}
\delta_P(d,d^\prime,a) & =P(d,a)+\frac{\gamma_{\mathrm{PT}}(f,d,d^{\prime},a)}{\APO(f,d,\infty,a)}\frac{\ATE(f,d,a)-\ATE(m,d,a)}{\APO(f,d,\infty,a)-\gamma_{\mathrm{PT}}(f,d,d^{\prime},a)}\\
 & +\frac{\gamma_{\mathrm{PT}}(f,d,d^{\prime},a)-\gamma_{\mathrm{PT}}(m,d,d^{\prime},a)}{\APO(f,d,\infty,a)-\gamma_{\mathrm{PT}}(f,d,d^{\prime},a)}.
\end{align*}
Rearranging terms we obtain the statement in the lemma. 
\end{proof}

\begin{proof}[Proof of Theorem \ref{thm:ntd_new_estimand}]
Substituting definitions, Assumption \ref{A.ntd_pt} can be written as 
\begin{align*}
 & \frac{\APO\left(f,d,\infty,a\right)-\APO\left(f,d,\infty,d-1\right)-\left[\APO\left(f,d^{\prime},\infty,a\right)-\APO\left(f,d^{\prime},\infty,d-1\right)\right]}{\APO\left(f,d,\infty,a\right)}\\
 & =\frac{\APO\left(m,d,\infty,a\right)-\APO\left(m,d,\infty,d-1\right)-\left[\APO\left(m,d^{\prime},\infty,a\right)-\APO\left(m,d^{\prime},\infty,d-1\right)\right]}{\APO\left(m,d,\infty,a\right)},
\end{align*}
which can be simplified into 
\begin{align}\label{eq:app_new_estimand_1}
 & \frac{\APO\left(f,d,\infty,d-1\right)+\left[\APO\left(f,d^{\prime},\infty,a\right)-\APO\left(f,d^{\prime},\infty,d-1\right)\right]}{\APO\left(f,d,\infty,a\right)} \nonumber\\
 & =\frac{\APO\left(m,d,\infty,d-1\right)+\left[\APO\left(m,d^{\prime},\infty,a\right)-\APO\left(m,d^{\prime},\infty,d-1\right)\right]}{\APO\left(m,d,\infty,a\right)}.
\end{align}
Substituting definitions into \eqref{eq:app_new_estimand_1} yields $\frac{\delta_{\mathrm{APO}}(f,d,d^{\prime},a)}{\delta_{\mathrm{APO}}(m,d,d^{\prime},a)}=\frac{\APO\left(f,d,\infty,a\right)}{\APO\left(m,d,\infty,a\right)}$ and consistency implies
$\frac{\APO\left(f,d,d,a\right)}{\APO\left(m,d,d,a\right)}=\frac{\mathbb{E}\left[Y_{a}\mid G=f,D=d\right]}{\mathbb{E}\left[Y_{a}\mid G=m,D=d\right]}$.
\end{proof}

\begin{proof}[Proof of Proposition \ref{prop:bias_correction}]
If $\APO(m,d,\infty,a)$ is known, \eqref{eq:app_proof_cp_ass_1} identifies
\[\gamma_{\mathrm{PT}}(m,d,d^{\prime},a) = \APO(m,d,\infty,a) - \delta_{\APO}(m,d,d^\prime,a).\]
Substituting into the bias parameter in Theorem~\ref{thm:cp_iden_res} yields
\[\mathrm{Bias}\left(d,d^{\prime},a\right)=\frac{\APO(m,d,\infty,a)}{\delta_{\APO}(m,d,d^{\prime},a)}.\]
Multiplying the result in Theorem~\ref{thm:cp_iden_res} by the inverse of the bias concludes the proof.
\end{proof}
\section{Estimation and Inference}
\label{sec:appendix_no_covariates}

For a triplet $(d,d',a)$, the main text considers the following descriptive estimands. Let $\mu_{g,d,a}=\mathbb{E}[Y_{i,a}\mid G_{i}=g,\,D_{i}=d]$ be the conditional mean of earnings. The descriptive DID estimands defined in \eqref{eq:desc_est} can be written as $\delta_{\APO}(g,d,d',a)=\mu_{g,d,d-1}+(\mu_{g,d',a}-\mu_{g,d',d-1})$, $\delta_{\ATE}(g,d,d',a)=\mu_{g,d,a}-\delta_{\APO}(g,d,d',a)$, and $\delta_{\theta}(g,d,d',a)=\delta_{\ATE}(g,d,d',a)/\delta_{\APO}(g,d,d',a)$. The three gender gaps are $\delta_{\Delta\ATE}(d,d',a)=\delta_{\ATE}(f,d,d',a)-\delta_{\ATE}(m,d,d',a)$, the gap in unnormalized DID; $\delta_{\Delta\theta}(d,d',a)=\delta_{\theta}(f,d,d',a)-\delta_{\theta}(m,d,d',a)$, the (biased) gap in normalized DID from Theorem~\ref{thm:cp_iden_res}; and $\delta_{\Delta\rho}(d,d',a)=\mu_{f,d,a}/\mu_{m,d,a}-\delta_{\APO}(f,d,d',a)/\delta_{\APO}(m,d,d',a)$, the new estimand from Theorem~\ref{thm:ntd_new_estimand}. Finally, applying Theorem~\ref{thm:ntd_new_estimand} together with $\APO(m,d,\infty,a)=\mu_{m,d,a}/(1+\theta(m,d,a))$ and equation~\eqref{eq:ntd_base_to_alt}, the bias-corrected gender gap from Proposition~\ref{prop:bias_correction} is
\[\delta_{\Delta\theta,\mathrm{BC}}(d,d',a \mid\theta(m,d,a))=\delta_{\Delta\rho}(d,d',a)\,\frac{\delta_{\APO}(m,d,d',a)}{\delta_{\APO}(f,d,d',a)}\,\bigl(1+\theta(m,d,a)\bigr).\]
Estimators replace each population mean with the sample mean $\bar{Y}_{g,d,a}=n_{g,d}^{-1}\sum_i S_i(g,d)\,Y_{i,a}$, where $S_i(g,d)=\mathbf{1}\{G_i=g,\,D_i=d\}$ and $n_{g,d}=\sum_i S_i(g,d)$.

Under regularity conditions and for a generic estimand $\delta$ and estimator $\widehat{\delta}$, $\sqrt{n}\,(\widehat\delta-\delta)=n^{-1/2}\sum_i\varphi_i+o_p(1)$ where $\varphi_i$ are individual components of a generic influence function (IF) $\psi_\delta$. Standard errors use cluster-robust variance, with clustering at the individual level following standard practice for difference-in-differences and panel data analyses \citep{bertrand2004much}. Replacing expectations with sample analogs: $\widehat{\mathrm{Var}}(\hat\delta)=n^{-2}\sum_c\bigl(\sum_{i\in\mathcal{I}_c}\hat\varphi_i\bigr)^2$, where $c$ indexes individuals and $\mathcal{I}_c$ is the set of observations belonging to individual $c$. The IF of the conditional mean follows from the moment-ratio identity $\mu_{g,d,a}=\mathbb{E}[S(g,d)Y_a]/\mathbb{E}[S(g,d)]$ and the quotient rule:
\[\psi_{\mu_{g,d,a}}=\frac{S(g,d)}{\mathbb{E}[S(g,d)]}\,(Y_a-\mu_{g,d,a}).\]
The IFs for $\delta_{\APO}$, $\delta_{\ATE}$, $\delta_{\Delta\ATE}$, and $\delta_{\Delta\theta}$ follow by linearity from the IFs above and below. The two non-trivial composite IFs (involving the quotient rule) are
{\small
\begin{align*}
\psi_{\delta_{\theta}(g,d,d',a)}
&= \frac{\psi_{\delta_{\ATE}(g,d,d',a)}}{\delta_{\APO}(g,d,d',a)}
   - \frac{\delta_{\ATE}(g,d,d',a)}{\delta_{\APO}(g,d,d',a)^{2}}\,\psi_{\delta_{\APO}(g,d,d',a)}, \\[4pt]
\psi_{\delta_{\Delta\rho}(d,d',a)}
&= \frac{\psi_{\mu_{f,d,a}}}{\mu_{m,d,a}}
   - \frac{\mu_{f,d,a}\,\psi_{\mu_{m,d,a}}}{\mu_{m,d,a}^{2}}
   - \frac{\psi_{\delta_{\APO}(f,d,d',a)}}{\delta_{\APO}(m,d,d',a)}
   + \frac{\delta_{\APO}(f,d,d',a)\,\psi_{\delta_{\APO}(m,d,d',a)}}{\delta_{\APO}(m,d,d',a)^{2}}.
\end{align*}
}

\subsection{Aggregate Estimator}\label{sec:agg_estimator}

This subsection extends the preceding discussion on estimation and inference from single treatment groups to multiple treatment groups.
For event time $e$, let $\mathcal{D}(e)=\{d:d+e<D_{\max}\}$ and $d'=d+e+1$. 
A generic aggregate causal estimand is identified by a descriptive aggregated estimand with the form $A(e)=\sum_{d\in\mathcal{D}(e)} w_d\,\delta_d$, 
where $\delta_d$ is the descriptive estimand of treatment group $d$ and $w_d=p_d/P_e$ is the population treatment-group share within $\mathcal{D}(e)$, with $p_d=\Pr(D=d)=\mathbb{E}[S(f,d)+S(m,d)]$, $P_e=\sum_{k\in\mathcal{D}(e)}p_k$, and $S(g,d)=\mathbf{1}\{G=g,\,D=d\}$ as above.  As an example aggregate, consider the population quantity $\delta_{\Delta\rho}$, which identifies $\Delta\rho$ under NTD (Section~\ref{sec:alt_estimand}). $A_{\delta_{\Delta\rho}}(e)=\sum_d w_d \, \delta_{\Delta\rho}(d,d',d+e)$ is an aggregate of the effect of parenthood on the gender earnings ratio at event time $e$ across treatment groups. For estimators of the weights I use the empirical treatment-group distribution. If both genders contribute to a single-treatment-group estimator, such as in the case of $\delta_{\Delta\rho}(d,d',d+e)$, the weights sum individuals of both genders, as in $w_{d}(e)=(n_{f,d}+n_{m,d})/\sum_{k\in\mathcal{D}(e)}(n_{f,k}+n_{m,k})$ where $n_{g,d}$ is the number of individuals from gender $g$ in treatment group $d$; for gender-specific estimands, the weights sum the corresponding gender only.

The IF of a generic $A(e)$ follows from linearity and the quotient rule applied to both $\delta_d$ and $w_d$. For example, abusing notation for brevity, write $\delta_{\Delta\rho}(d,e) \equiv \delta_{\Delta\rho}(d,d',d+e)$ and $A(e)\equiv A_{\delta_{\Delta\rho}}(e)$. Applying IF algebra gives
\[
\psi_{A(e)}=\tfrac{1}{P_{e}}\sum_{d\in\mathcal{D}(e)}\bigl[(\delta_{d}-A(e))\psi_{p_{d}}+p_{d}\psi_{\delta_{d}}\bigr],\quad \psi_{p_{d}}=S(f,d)+S(m,d)-p_{d}.
\]

All descriptive quantities above are combinations of conditional outcome means and the treatment-group probability. 
Estimators and standard errors are obtained by replacing each population expectation with its sample analog. 
As above, variance estimators for the aggregated estimators likewise cluster by individuals.

\clearpage
\setcounter{figure}{0}
\setcounter{table}{0}
\section{Aggregate Estimates for Israel, the UK, and Germany}\label{sec:agg_estimates_app}

Figure~\ref{fig:agg_conv_vs_new} reports estimates of $\rho_{\mathrm{Agg}}(e)$ from \eqref{eq:agg_rho} and $\Delta\theta_{\mathrm{Agg}}(e)$ from \eqref{eq:agg_conv}, by event time. The sum of female and male treated units is used as weights for each treatment group, and each treatment-group-specific causal estimand is replaced by its $2\times2$ descriptive counterpart, estimated via sample analogs. For further discussion of the aggregate estimator, see Appendix~\ref{sec:appendix_no_covariates}. Panels report results for Israel, the UK, and Germany.

Across all three countries, $\rho_{\mathrm{Agg}}(e)$ is less negative than $\Delta\theta_{\mathrm{Agg}}(e)$ at every event time. For example, at five years post-childbirth, $\Delta\theta_{\mathrm{Agg}}(5)=-17.4\%$ while $\rho_{\mathrm{Agg}}(5) = -13.8\%$ for Israel, a difference of $3.6$ percentage points ($21\%$). These patterns are consistent with the discussion on single-treatment-group estimates in Section~\ref{sec:new_ntd_empirical}.

\begin{figure}[t!]
    \centering
    \includegraphics[width=\textwidth]{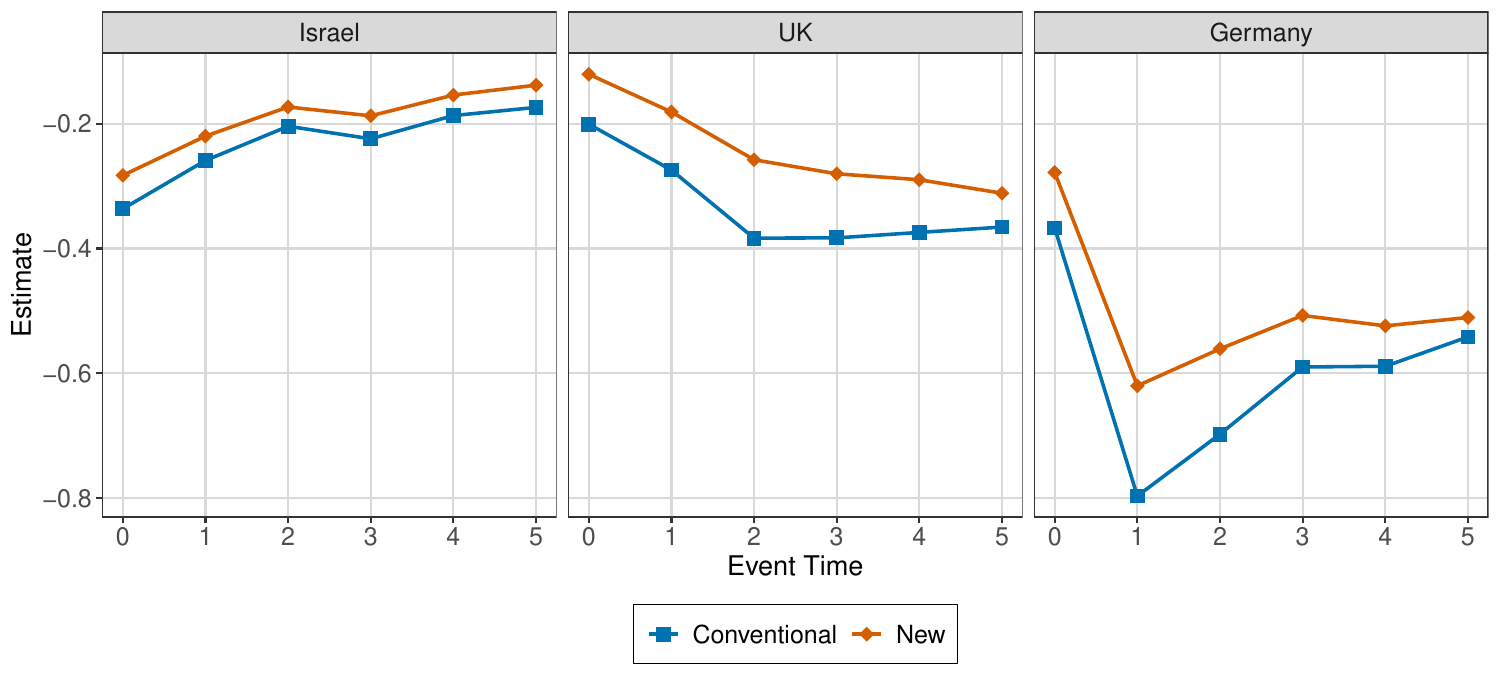}
    \caption{Aggregated Estimates: Conventional vs.\ New Estimator. \textit{Notes:} The figure compares the conventional aggregate gender gap $\Delta\theta_{\mathrm{Agg}}(e)$ (blue squares) with the new aggregate $\rho_{\mathrm{Agg}}(e)$ (orange diamonds), by event time, separately for Israel, the UK, and Germany. Single-treatment-group estimates are weighted by the sum of female and male units in each treatment group.}
    \label{fig:agg_conv_vs_new}
\end{figure}

\end{document}